\documentclass[conference]{IEEEtran}
\IEEEoverridecommandlockouts

\usepackage{cite}
\usepackage{amsmath,amssymb,amsfonts}
\usepackage{algorithmic}
\usepackage{graphicx}
\usepackage{textcomp}
\usepackage{xcolor}
\usepackage{algorithm}
\usepackage{amsthm}
\usepackage{bm}
\usepackage{url}

\newtheorem{theorem}{Theorem}
\newtheorem{assumption}{Assumption}
\newtheorem{lemma}{Lemma}
\newtheorem{corollary}{Corollary}

\begin{document}

\title{Optimizing Information Freshness for Wireless Local Area Networks with Multiple APs}

\author{\IEEEauthorblockN{Ananth Ram Rajagopalan\textsuperscript{*}}
\IEEEauthorblockA{\textit{Purdue University}\\
West Lafayette, Indiana, USA \\
rajago19@purdue.edu}
\and
\IEEEauthorblockN{Jiahui Ni\textsuperscript{*}}
\IEEEauthorblockA{\textit{Purdue University}\\
West Lafayette, Indiana, USA \\
ni169@purdue.edu}
\and
\IEEEauthorblockN{Vishrant Tripathi}
\IEEEauthorblockA{\textit{Purdue University}\\
West Lafayette, Indiana, USA \\
tripathv@purdue.edu}
\thanks{\textsuperscript{*}Both authors contributed equally to this research.}
}

\maketitle

\begin{abstract}
Dense indoor WLANs increasingly rely on multiple access points (APs) operating over partially overlapping spectrum to support latency-sensitive applications. In such deployments, simultaneous transmissions across APs create co-channel and adjacent-channel interference, making scheduling decisions interdependent and directly impacting information freshness. Motivated by emerging software-defined WLAN architectures that enable centralized coordination, we study the problem of minimizing network-wide Age of Information (AoI) in multi-AP WLANs. Unlike classical AoI scheduling that runs at a single AP, each scheduling decision is now coupled across APs due to interference. This leads to a new class of combinatorial AoI control problems with action-dependent time evolution. We first derive a lower bound on the achievable AoI under arbitrary scheduling policies. We then design stationary randomized policies that have constant-factor optimality guarantees relative to this bound. Building on these insights, we develop a Lyapunov drift-based online policy for systems with action-dependent frame lengths, and establish constant-factor guarantees using new ratio-based drift analysis. To enable scalable implementation, we further show that per-frame scheduling admits  efficient polynomial-time local-search approximations under a submodularity assumption. Simulations using realistic WLAN layouts demonstrate about 50\% AoI reduction over distributed single AP baselines.
\end{abstract}


\section{Introduction}

Modern indoor wireless local area networks (WLANs) increasingly rely on dense deployments of WiFi access points (APs) to provide ubiquitous high-throughput connectivity. Such environments are common in enterprise buildings, factories, warehouses, airports, and smart retail spaces, where multiple APs operate over partially overlapping spectrum. Beyond traditional best-effort traffic, these networks are now expected to support latency-sensitive real-time applications such as industrial automation\cite{nokia_warehouse_automation_private_wireless, warehouse}, connected multi-robot swarms \cite{uav1,uav2}, immersive augmented/virtual reality \cite{braud2017future,jansen2023can}, and large-scale IoT monitoring \cite{iot1,iot2}. In these applications, timely state updates are critical for control, coordination, and user interactivity.

Despite dense infrastructure deployment, wireless coordination in current WLANs remains largely decentralized. Existing network planning methods focus primarily on AP placement and channel assignment, while medium access is delegated to distributed CSMA/CA mechanisms in IEEE 802.11. However, simultaneous transmissions across neighboring APs often generate co-channel and adjacent-channel interference, while physical obstructions such as walls, shelving units, or metallic racks further aggravate hidden-node effects. As a result, independent contention-based medium access can lead to collisions, large delays, and inefficient spectrum usage, limiting support for real-time traffic.

To quantify timeliness, the Age of Information (AoI) has emerged as a standard performance metric. AoI measures the time elapsed since the freshest successfully received update was generated \cite{UpdateOrWait}. Prior work has looked at AoI optimization in a variety of  queueing systems and wireless network settings \cite{queue1,queue2,aoi1,aoi3,whittle,interference,LargeAndSmall}. However, most existing work on AoI assumes a single centralized scheduler and simplified (binary) interference constraints. They do not capture typical WLAN deployments where multiple APs simultaneously serve users. For example, adding a user to the active set might reduce the transmission rate for both users but might still be a net gain for network AoI. Coordination among APs can thus lead to performance gains. However, achieving coordinated scheduling efficiently and optimizing AoI in settings where the set of scheduled transmissions jointly determines both update deliveries and completion times is an open question.

This paper studies \textbf{network-wide AoI minimization in multi-AP WLANs with centralized coordination}. Motivated by the emergence of software-defined networking (SDN) and its applications to wireless scenarios \cite{sdn1,sdn2,sdnwlan,sdnwlan2,mininet-wifi}, we consider a controller that observes network state and schedules transmissions across multiple APs. We assume AP locations and channel assignments are fixed inputs determined by existing planning tools, and focus exclusively on the scheduling problem. Unlike classical slotted AoI formulations where packet deliveries always take a fixed amount of time, each scheduling action in our system determines both \emph{which updates are delivered} and \emph{how much time elapses}, since scheduling frame lengths become a function of SINR-dependent transmission rates. This creates a new class of AoI problems with action-dependent time evolution and combinatorial coupling across APs.

\textbf{Related Work:} There has been recent interest in the WLAN research literature on coordinated operation across multiple APs to improve performance. Surveys such as \cite{verma2023survey} identify coordinated scheduling, interference management, and joint transmission as key directions for next-generation WiFi. Recent works further demonstrate \cite{nunez2025enabling, imputato2024beyond, mediatek_anti_interference_wifi_2025} throughput and spatial reuse gains through multi-AP joint scheduling and interference management in Wi-Fi 8 and beyond, while standardization-oriented works for \textit{IEEE 802.11be} similarly emphasize AP coordination and infrastructure-assisted scheduling \cite{yang2019ap}. Along parallel lines, recent works have started demonstrating the practical feasibility of software-defined WLAN architectures  \cite{sdn1,sdn2,sdnwlan,sdnwlan2,mininet-wifi,vatsa2025optimum} at small scale. As discussed above, there is plenty of work in AoI literature on optimization for wireless networks \cite{aoi1,aoi3,whittle,interference,LargeAndSmall}, including scheduling under general binary interference constraints \cite{interference}. See \cite{yates2021age,sun2022age} for excellent surveys on AoI.

Crucially, prior multi-AP WLAN works focus primarily on throughput and spatial reuse, while existing AoI works typically assume a single scheduler or fixed slot durations and simplified interference models. Closest to our work, the authors of \cite{interference} study optimization under binary interference constraints. However, they restrict analysis to a limited class of policies (stationary randomized) and all packets take exactly one time-slot for delivery. In this work, we study freshness-optimal coordinated scheduling in dense multi-AP WLANs with SINR-dependent general interference constraints, where concurrent transmissions may succeed at reduced rates and each scheduling decision jointly determines both service outcomes and frame durations. Our \textbf{key contributions} are summarized below.
\begin{itemize}
    \item \textbf{Fundamental Limits:} We derive a lower bound on achievable weighted sum AoI under arbitrary scheduling policies in multi-AP WLANs.
    
    \item \textbf{Stationary Randomized Policies:} We design stationary randomized scheduling policies, obtain closed-form expressions for the resulting weighted sum AoI, and prove constant-factor optimality guarantees relative to the lower bound.
    
    \item \textbf{Lyapunov Policy with Action-Dependent Frame Lengths:} We develop a Lyapunov drift-based policy for systems where frame durations depend on the chosen scheduling actions. We introduce a new ratio-based drift analysis and establish constant-factor performance guarantees for the resulting online policy.
    
    \item \textbf{Scalable Optimization via Submodularity:} Under a submodularity assumption, we show that per-frame max-weight scheduling can be implemented efficiently in polynomial time via local-search approximations instead of exhaustive combinatorial search.
    
    \item \textbf{Evaluation:} Simulations using realistic WLAN layouts demonstrate about \textbf{$\mathbf{50}$\% AoI reductions} over decentralized per-AP scheduling baselines.
\end{itemize}

While we focus on multi-AP coordination in WLANs, the mathematical framework we develop can also be applied to scheduling in cellular networks to mitigate inter-cell interference and handle non-orthogonal multiple access (NOMA). The remainder of this paper is organized as follows. In Section~\ref{sec:model}, we present our system model and the multi-AP AoI optimization problem. In Section~\ref{sec:lower_bound}, we derive a lower bound on achievable weighted sum AoI. In Section~\ref{section:srp}, we derive optimal stationary randomized policies and prove performance bounds. In Section~\ref{section:mw}, we develop a novel Lyapunov scheduling policy and a scalable implementation via submodularity. Section~\ref{sec:simulations} contains our simulation results.

\section{System Model}
\label{sec:model}
\begin{figure}[t]
    \centering
    \includegraphics[width=0.5\textwidth]{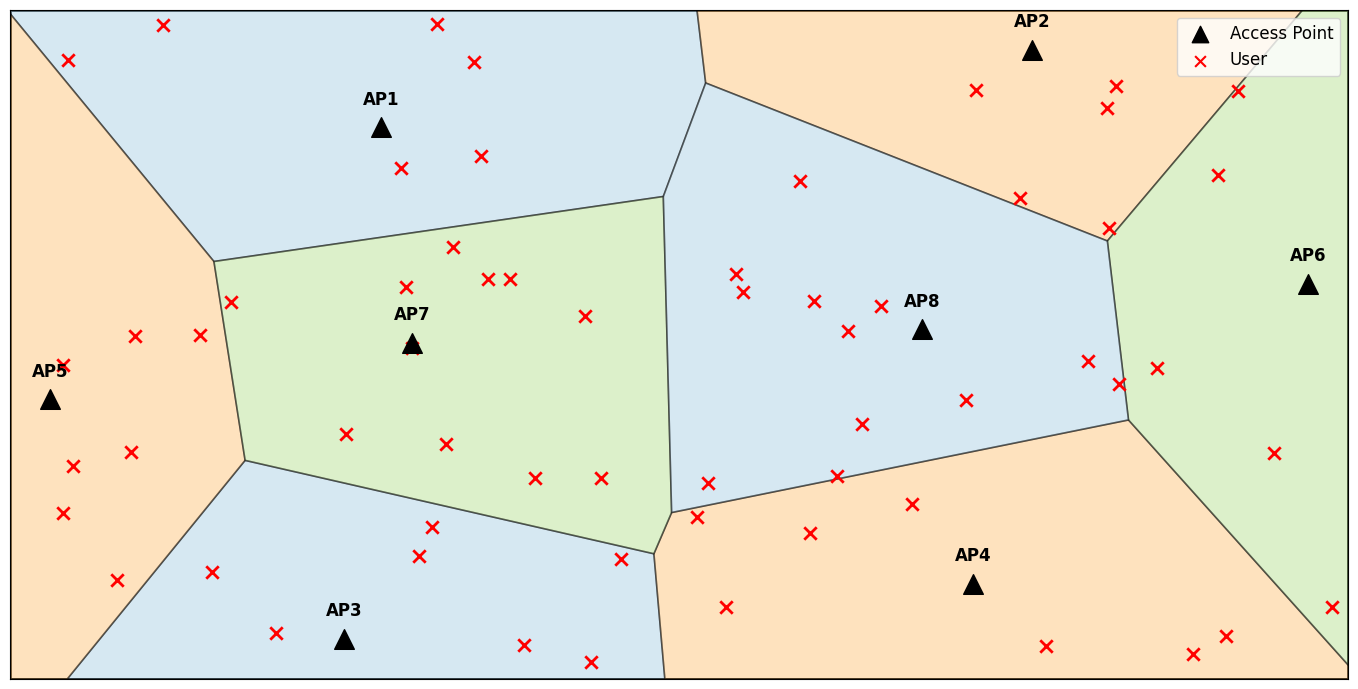}
    \caption{A WLAN with $K=8$ APs and $N=60$ users. There are three orthogonal channels available for use, so APs on the same channel, for example (1,3,8), have interfering transmissions. \vspace{-3mm}}
    \label{fig:ap_fig}
\end{figure}

We consider an indoor wireless local area network (WLAN) consisting of $K$ access points (APs) serving $N$ wireless nodes, as illustrated in Fig.~\ref{fig:ap_fig}. Each node associates with the geographically closest AP. This induces a partition of nodes into sets $\mathcal{N}_k$, where $\mathcal{N}_k$ denotes the nodes served by AP $k$. The AP locations and channel assignments are fixed and assumed to be known to the scheduler.

In conventional WLANs, nodes access the medium using distributed MAC protocols such as CSMA/CA. While effective within a single cell, these mechanisms do not coordinate transmissions across multiple APs. Modern enterprise and industrial WLANs promise centralized controllers via software defined networking (SDN) that are capable of coordinating transmissions network-wide. Our goal is to design such a centralized scheduling policy to optimize information freshness across all nodes.

\subsection{Transmission Times}

We assume that time is slotted and that each node can generate status updates of fixed size ($L$ bits) at will. Due to local bandwidth and interference constraints, at most one node associated with each AP may transmit at any given time. Let $S(t)$ denote the set of nodes transmitting during slot $t$. Then $S(t)$ satisfies
\begin{equation}
|S(t)\cap \mathcal{N}_k| \le 1, \quad \forall k=1,\dots,K.
\end{equation}

Although only one node per AP transmits simultaneously, concurrent transmissions across different APs may interfere depending on spatial layout and channel reuse.

A key feature of practical WLAN deployments is limited spectrum availability. Neighboring APs can often operate on partially overlapping frequency channels rather than fully orthogonal ones. Consequently, simultaneous transmissions across APs influence each other's achievable data rates. If fully orthogonal channels were available, scheduling decisions would decompose into $K$ independent single-AP problems, which are well studied in literature. In contrast, spectrum overlap creates coupling across APs, making coordinated scheduling essential.

To capture these effects, we model the transmission duration of node $i$ as a function of the transmission set, i.e. $\tau_i(S(t)) \in \mathbb{R}^+$. Here, $\tau_i(S(t))$ represents the time required for node $i$ to complete transmission when the active set is $S(t)$ and $i\in S(t)$. This abstraction incorporates interference, channel assignments, node and AP locations, propagation effects, and environmental geometry.

\paragraph{Example (SINR-based model).}
As an illustrative instance, suppose each node transmits with power $P_0$ under path-loss propagation. The SINR of node $i$ communicating with AP $k$ under activation set $S(t)$ is
\[
\text{SINR}_i(S(t))
= \frac{P_0 d_{ik}^{-\alpha}}
{N_0B + \sum_{j\in S(t), j\neq i}\eta_{ij} P_0 d_{jk}^{-\alpha}},
\]
where $d_{ik}$ is the distance between node $i$ and AP $k$, $\alpha$ is the path-loss exponent, $B$ is bandwidth, and $\eta_{ij}\in[0,1]$ represents spectral overlap between transmissions of nodes $i$ and $j$. The resulting transmission time is
\[
\tau_i(S(t))
= \bigg\lceil \frac{L/T_s}{B\log_2(1+\text{SINR}_i(S(t)))} \bigg\rceil,
\]
where $T_s$ is the time-slot duration and $\tau_i(S(t))$ is measured in number of time-slots.

While this example uses a simplified physical-layer model, our analysis treats $\tau_i(S)$ as a general measurable function, allowing incorporation of fading, obstacles, adaptive modulation, different update sizes, and realistic WLAN behavior. Note that binary interference constraints with unit rate links, as studied in \cite{interference}, are a special case of the general system model we analyze in our work.

\subsection{Age of Information}\label{section2.2}

Let $A_i(t)$ denote the Age of Information (AoI) of node $i$ at time $t$, defined as the time elapsed since the most recently received update was generated. In realistic WLAN operation, nodes within an activation set complete transmissions at different times, causing asynchronous age resets. Exact modeling therefore leads to a stochastic control problem in which scheduling decisions influence both future ages and future completion times. The resulting system becomes a high-dimensional semi-Markov decision process that is analytically intractable.

To obtain tractable structural insights, we adopt a frame-based abstraction. We describe system evolution through frames indexed by $b=1,2,\dots$. Frame $b$ starts at time $t_b$, and at the beginning of it the controller selects a transmission set $S_b$. The following assumptions hold:
\begin{itemize}
\item The activation set remains fixed during the frame.
\item The controller waits until all nodes in $S_b$ finish transmitting before selecting the next set.
\item Channels are reliable; although our approach allows for an easy extension to the unreliable channel case.
\end{itemize}

The duration of frame $b$ is determined by the slowest transmission:
\begin{equation}
\Delta(S_b) = \max_{i\in S_b} \tau_i(S_b).
\end{equation}
Thus, a single slow transmission can delay the completion of all concurrent updates, coupling the age evolution of all nodes. AoI evolves as follows over time-slots
\begin{equation}
\label{eq:aoi_evolution}
A_i(t+1)=
\begin{cases}
\Delta(S_b), & \text{if } t+1=t_b+\Delta_b \text{ and } i\in S_b,\\
A_i(t)+1, & \text{otherwise}.
\end{cases}
\end{equation}
The reset value $\Delta(S_b)$ reflects that updates generated at the start of the frame are already $\Delta(S_b)$ slots old upon delivery.

This abstraction enables tractable analysis and policy design while retaining the key coupling induced by interference-dependent transmission durations. Note that our assumption of synchronous delivery necessarily provides a sample path upper bound for the AoI of any node given any sequence of scheduling decisions. Thus, using this frame based model still leads to AoI analysis and policies that perform well under asynchronous packet completions. We verify this with simulations in Sec.~\ref{sec:simulations}, where we show that the closed form expressions for AoI we derive closely match the network average AoI even when packet completions are asynchronous.

\subsection{Optimization Problem}
Our goal is to minimize the long-term weighted average AoI:
\begin{equation}
\label{eq:aoi_optimization}
\boxed{
\begin{aligned}
\min_{\pi\in\Pi}\;
& \lim_{T\to\infty} \frac{1}{N} \frac{1}{T}
\sum_{t=1}^{T}\sum_{i=1}^{N} w_i\,\mathbb{E}[A_i(t)]\\
\text{s.t. }&
|S(t)\cap\mathcal{N}_k|\le1,
\quad \forall k,t,
\end{aligned}}
\end{equation}
where $\Pi$ denotes the class of causal scheduling policies.

Solving~\eqref{eq:aoi_optimization} is challenging for three reasons. First, the scheduler must search over a combinatorial action space of size $\prod_{k=1}^{K}(|\mathcal{N}_k|+1)$ in each time-slot. So an optimal policy found using dynamic programming will suffer from the curse of dimensionality very quickly as the number of APs $K$ increases. Second, inter-AP interference couples transmission durations across cells so the problem cannot be easily decoupled to single AP formulations that are solved in prior works. Third, scheduling decisions determine the frame duration $\Delta_b$, thereby altering the time scale and rendering standard Lyapunov analysis not directly applicable.

In the rest of the paper, we will show how to overcome these challenges and design policies that can solve \eqref{eq:aoi_optimization} in a computationally efficient manner. To do so, we first derive a lower bound on the achievable AoI.

\section{Lower Bound}
\label{sec:lower_bound}
In this section, we will derive a lower bound on the achievable long-term weighted average AoI under any admissible scheduling policy $\pi \in \Pi$ for the problem \eqref{eq:aoi_optimization}. This allows us to prove performance guarantees for our proposed policies by comparing them to the lower bound. To do so, we start by defining two useful quantities. 

\textbf{Waiting Time and Service Time:} Consider a sample path $\omega \in \Omega$ associated with a scheduling policy $\pi \in \Pi$ and a time horizon $T$. Let $D_{i}(T)$ denote the total number of delivered packets from node $i$ to base station $k$ by the end of slot $T$ and let $m\in \left\{ 1,...,D_{i}(T)\right\}$ be the index of the delivered packets from node $i$. We define the waiting time of the $m$th packet from node $i$ as the number of time slots between the completion of transmission of the ($m-1$)th packet from node $i$ and the start of transmission of the $m$th packet from node $i$, and denote it by $W_{i}[m]$. We define the service time of the $m$th packet from node $i$ as the time required for the $m$th packet from node $i$ to complete transmission, and denote it by $S_{i}[m]$. 

The waiting time reflects the impact of the inter-scheduling interval of node $i$ on its AoI, and the service time captures the impact of the transmission time. The inter-delivery time between the ($m-1$)th and the $m$th packets from node $i$ is given by $W_{i}[m]+S_{i}[m]$, during which the AoI of node $i$ increases linearly. We can use these two quantities to characterize the evolution of the AoI of node $i$ over the entire time horizon. Between the time slot after the ($m-1$)th packet finishes transmission and the time slot at which the $m$th packet finishes transmission, $A_{i}(t)$ evolves as $S_{i}[m-1],S_{i}[m-1]+2,...,S_{i}[m-1]+W_{i}[m]+S_{i}[m]-1$. Let $R_{i}$ denote the number of remaining time slots after the last packet delivery from node $i$. During these time slots, $A_{i}(t)$ evolves as $S_{i}[D_{i}(T)]+1,...,S_{i}[D_{i}(T)]+R_{i}$.

Define the operator $\mathbb{\bar{M}}[\cdot]$ that calculates the sample mean of a set of values. Using this operator, the sample mean of $W_{i}[m]$ and $S_{i}[m]$ is given by:
\begin{equation} \label{eq:6}
    \mathbb{\bar{M}}[W_{i}[m]] = \frac{1}{D_{i}(T)}\sum_{m=1}^{D_{i}(T)}W_{i}[m],
\end{equation}
\begin{equation} \label{eq:7}
    \mathbb{\bar{M}}[S_{i}[m]] = \frac{1}{D_{i}(T)}\sum_{m=1}^{D_{i}(T)}S_{i}[m].
\end{equation}
We can then derive an explicit expression of the long-term weighted average AoI in terms of $W_{i}[m]$ and $S_{i}[m]$.
\begin{lemma}
\label{lemma:J_pi}
    The infinite-horizon weighted average AoI achieved by policy $\pi$, denoted as $J^{\pi}$, can be written as
    \begin{equation}\label{eq:11}
\begin{split}
    J^\pi = & \lim_{T\rightarrow \infty}\frac{1}{N}\sum_{i = 1}^{N}w_{i}\left[\frac{\mathbb{\bar{M}}[(W_{i}[m]+S_{i}[m])^2]}{2\mathbb{\bar{M}}[W_{i}[m]+S_{i}[m]]}\right.\\
        &\left.+\frac{\mathbb{\bar{M}}[S_{i}[m-1](W_{i}[m]+S_{i}[m])]}{\mathbb{\bar{M}}[W_{i}[m]+S_{i}[m]]}-\frac{1}{2}\right].
\end{split}
\end{equation}
\end{lemma}
\begin{proof}
See Appendix \ref{proof:lemma_1}.
\end{proof}

The first term on the RHS of Equation \eqref{eq:11} is the ratio of the sample second moment of the inter-delivery time over the sample first moment of the inter-delivery time. Let's define the operator $\mathbb{\bar{V}}[\cdot]$ that calculates the sample variance of a set of values. The first term of the RHS of Equation \eqref{eq:11} can then be written as $\mathbb{\bar{M}}[W_{i}[m]+S_{i}[m]]+\frac{\mathbb{\bar{V}}[W_{i}[m]+S_{i}[m]]}{2\mathbb{\bar{M}}[W_{i}[m]+S_{i}[m]]}$. This is the sum of the average inter-delivery time and a measure of the dispersion of the inter-delivery time, which implies that minimizing the long-term weighted average AoI requires scheduling policies that not only reduce the mean inter-delivery time but also ensure regular and evenly spaced packet deliveries, which is a well known observation in AoI literature. The second term of the RHS of Equation \eqref{eq:11} depends on the service time of the previous packet and captures the effect of time-varying, interference-dependent transmission times on the long-term weighted average AoI. When the transmission time has a fixed length of one slot for all active sets, Equation \eqref{eq:11} reduces to the known result in \cite{throughput}. 

Lemma \ref{lemma:J_pi} gives us an expression of the long-term weighted average AoI of an arbitrary admissible scheduling policy $\pi \in \Pi$. Applying Jensen's Inequality and the fact that $S_{i}[m-1](W_{i}[m]+S_{i}[m])>0$ to Equation \eqref{eq:11}, we can obtain a lower bound for the performance of all admissible scheduling policies.

We define the long-term throughput of node $i$ under policy $\pi$ as the total number of delivered packets over the time horizon $T$ as $T\rightarrow \infty$. This can be expressed as:
\begin{equation}\label{eq:throughput_1}
    q_{i}^\pi =\lim_{T\rightarrow\infty}\frac{D_{i}(T)}{T}.
\end{equation}
\begin{theorem}
\label{thm:lb}
    The optimization problem in (\ref{eq:LB_opt}) provides a lower bound $L_B$ to the AoI optimization problem.
    \begin{equation} \label{eq:LB_opt}
    \begin{split}
        L_B = &\min_{\pi\in\Pi}\frac{1}{N}\sum_{i = 1}^{N}w_{i}\left(\frac{1}{2q_{i}^\pi}-\frac{1}{2}\right)\\
        \text{s.t.}\;&
|S(t)\cap\mathcal{N}_k|\le1,
\quad \forall k,t.
    \end{split}
\end{equation}
\end{theorem}
\begin{proof}
Applying Jensen's Inequality and $S_{i}[m-1](W_{i}[m]+S_{i}[m])>0$ yields
\begin{equation}\label{eq:jpi}
\begin{split}
    J^\pi > & \lim_{T\rightarrow \infty}\frac{1}{N}\sum_{i = 1}^{N}w_{i}\left[\frac{\mathbb{\bar{M}}[(W_{i}[m]+S_{i}[m])]}{2}-\frac{1}{2}\right].
\end{split}
\end{equation}
Substituting (\ref{eq:8}) into (\ref{eq:throughput_1}) yields
\begin{equation}\label{eq:throughput_2}
\begin{split}
    q_{i}^\pi &=\lim_{T\rightarrow\infty}\frac{1}{\mathbb{\bar{M}}[W_{i}[m]]+\mathbb{\bar{M}}[S_{i}[m]]+\frac{R_{i_k}}{D_{i}(T)}}\\
    &=\frac{1}{\mathbb{\bar{M}}[W_{i}[m]]+\mathbb{\bar{M}}[S_{i}[m]]}.
\end{split}
\end{equation}
Substituting (\ref{eq:throughput_2}) into (\ref{eq:jpi}) yields
\begin{equation}\label{eq:jpi_lb}
\begin{split}
    J^\pi > & \frac{1}{N}\sum_{i = 1}^{N}w_{i}\left(\frac{1}{2q_{i}^\pi}-\frac{1}{2}\right).
\end{split}
\end{equation}
\end{proof}
Equation \eqref{eq:LB_opt} is consistent with the lower bound in the fixed single-slot transmission time case derived in \cite{throughput}. This follows from the fact that the second term on the RHS of Equation \eqref{eq:11}, which accounts for the impact of the time-varying, interference-dependent transmission times, has been omitted in the derivation of the lower bound. This could have led to a loose lower bound $L_B$, but in the next section, we show that a stationary randomized policy yields a tight optimality ratio with respect to this lower bound. 

Equation \eqref{eq:LB_opt} also aligns with the intermediate result in \cite{LargeAndSmall}, in which the authors consider updates consisting of multiple packets delivered over multiple time slots, and the AoI is updated only when all packets in an update are delivered. In \cite{LargeAndSmall}, the authors derive an analytical expression for the lower bound; however, obtaining such a closed-form expression is challenging in our setting, where $q_{i}^\pi$ depends on scheduling decisions that are coupled across all nodes in the transmission set. Nevertheless, the long-term throughput that solves \eqref{eq:LB_opt} can be achieved by a stationary randomized policy, which we will derive in the following section. Substituting this expression into \eqref{eq:LB_opt} allows the lower bound to be computed numerically.

\section{Stationary Randomized Policies} \label{section:srp}
Let $\Pi^{sr}$ denote the class of stationary randomized policies. The centralized scheduler running a stationary randomized policy selects a subset of nodes $S$ to transmit at the beginning of each frame with a probability $\mu_{S} \in [0,1]$. These probabilities are fixed for the entire time horizon. Recall that the subset $S$ must satisfy the constraint that at most one node can transmit to each AP, i.e. $\mathcal{I}=\left\{S\subseteq \mathcal{N}:|S\bigcap \mathcal{N}_k|\leq1,\forall k \in \left\{1,...,K\right\}\right\}$. The probabilities satisfy the condition that $\sum_{S\in \mathcal{I}}\mu_{S}\leq1.$

Then each policy $R \in \Pi^{sr}$ is fully described by the set of probabilities $\left\{\mu_{S}|S\in\mathcal{I}\right\}$. The set-wide transmission completion time is given by $\Delta(S)$. Theorem \ref{thm:srp_ewsaoi} relates the long-term weighted average AoI to the scheduling probabilities $\left\{\mu_{S}|S\in\mathcal{I}\right\}$ and the set-wide transmission completion time $\Delta(S)$.

\begin{theorem}
    \label{thm:srp_ewsaoi}
    Consider any stationary randomized policy R with scheduling probabilities $\left\{\mu_{S}|S\in\mathcal{I}\right\}$. The long-term weighted average AoI is given by
    \begin{equation} \label{eq:J_R_thm}
    \begin{split}
        \mathbb{E}[J^R] &= \frac{1}{N}\sum_{i = 1}^{N}w_{i}\Bigg(\frac{\sum_{S\in \mathcal{I}}\Delta(S)^2\mu_{S}}{2\sum_{S\in \mathcal{I}}\Delta(S)\mu_{S}} \\
        &\quad +\frac{\sum_{S\in \mathcal{I}}\Delta(S)\mu_{S}}{\sum_{ \{S\in \mathcal{I}: i\in S \}}\mu_{S}}-\frac{1}{2}\Bigg).
    \end{split}
    \end{equation}
\end{theorem}
\begin{proof}
    See Appendix \ref{proof:thm4.1}.
\end{proof}

Equation \eqref{eq:J_R_thm} reduces to the known result in \cite{aoi1} and \cite{throughput} when $\Delta(S) = 1, \forall S\in\mathcal{I}$. Using the above theorem, we can find the optimal scheduling probabilities $\left\{\mu_{S}^*|S\in\mathcal{I}\right\}$ by solving the following optimization problem:
\begin{equation}\label{eq:J_R_opt}
\min_{\left\{\mu_{S}|S\in\mathcal{I}\right\}}\mathbb{E}[J^R],\quad \text{s.t.} \sum_{S\in \mathcal{I}}\mu_{S}\leq1,\mu_S \geq0, \forall S\in\mathcal{I}
\end{equation}
The optimization problem \eqref{eq:J_R_opt} can be converted into a convex optimization problem using a change of variables.
\begin{lemma}\label{lemma:J_R_cvx}
    Solving the optimization problem \eqref{eq:J_R_opt} is equivalent to solving the following convex optimization problem:
    \begin{equation}\label{eq:J_R_opt_cvx}
    \begin{aligned}
        \min_{\mathbf{x},c}\;&\frac{1}{N}\sum_{i=1}^Nw_i\left(\mathbf{a}^{\top}\mathbf{x}+\frac{1}{2\bm{M}_{i}\mathbf{x}}\right)\\
        \text{s.t.} \;&\mathbf{1}^{\top}\mathbf{x} \leq c,\\
        \;&2\mathbf{b}^{\top}\mathbf{x}+2c=1,\\
        \;&\mathbf{x}\geq 0        
    \end{aligned}
    \end{equation}
    where $\mathbf{a}, \mathbf{b}, \mathbf{x} \in \mathbb{R}^{|\mathcal{I}|}$, and $\bm{M} \in \mathbb{R}^{N \times |\mathcal{I}|}$, with $\mathbf{a} = (\Delta(S)^2)_{S\in\mathcal{I}}$, $\mathbf{b} = (\Delta(S))_{S\in\mathcal{I}}$, and $M_{i,S} = 1$ if $i \in S$, and $0$ otherwise. A solution $\mathbf{x}$ to the convex optimization problem \eqref{eq:J_R_opt_cvx} can be translated to a solution of the original optimization problem \eqref{eq:J_R_opt} via $\mu_S = \frac{x_S}{c}, \forall S\in\mathcal{I}$.
\end{lemma}
\begin{proof}\label{proof:lemma_2}
Let $\bm{\mu} = (\mu_S)_{S \in \mathcal{I}} \in \mathbb{R}^{|\mathcal{I}|}$ denote the vector of probabilities over $\mathcal{I}$. The optimization problem in \eqref{eq:J_R_opt} can be written as
\begin{equation}\label{eq:J_opt_vec}
    \begin{aligned}
        \min_{\bm{\mu}}\;&\frac{1}{N}\sum_{i = 1}^{N}w_{i}\left(\frac{\mathbf{a}^{\top}\bm{\mu}}{2\mathbf{b}^{\top}\bm{\mu}}+\frac{\mathbf{b}^{\top}\bm{\mu}}{\bm{M}_{i}\bm{\mu}}-\frac{1}{2}\right),\\
        \text{s.t.}\;&\mathbf{1}^{\top}\bm{\mu}\leq1,\\
        &\bm{\mu}\geq0,
    \end{aligned}
\end{equation}
where $\bm{M}_{i}$ is the row vector corresponding to the $i$th row of $\bm{M}$. Let $\mathbf{x} = \frac{\bm{\mu}}{2\mathbf{b}^{\top}\bm{\mu}}$, $c = \frac{1}{2\mathbf{b}^{\top}\bm{\mu}}.$ Substituting $\mathbf{x}$ and $c$ into \eqref{eq:J_opt_vec} and omitting the constant term yields \eqref{eq:J_R_opt_cvx}.
\end{proof}

Lemma \ref{lemma:J_R_cvx} establishes that the optimization problem in \eqref{eq:J_R_opt} can be reformulated as an equivalent convex optimization problem. As a result, a globally optimal solution to \eqref{eq:J_R_opt} can be computed using standard convex optimization solvers such as CVX \cite{cvx}. The optimal solution yields the set of scheduling probabilities $\left\{\mu_{S}|S\in\mathcal{I}\right\}$ that specifies the optimal stationary randomized policy. We define the optimality ratio of the stationary randomized policy $R$ as $\rho^R = \mathbb{E}[J^R]/L_B$.
\begin{theorem}\label{thm:opt_ratio}
    Consider an optimal stationary randomized policy $\pi^*_{SR}$ defined by the scheduling probabilities $\left\{\mu_{S}^*|S\in\mathcal{I}\right\}$ that is a solution to (\ref{eq:J_R_opt}). Let $\pi_{LB}$ denote the stationary randomized policy defined by the scheduling probabilities $\left\{\mu_{S}^{LB}|S\in\mathcal{I}\right\}$ that solves the Lower Bound optimization problem in (\ref{eq:LB_opt}). The optimality ratio of $\pi^*_{SR}$ is such that
    \begin{equation}
    \label{eq:opt_sr}
    \begin{split}
        &\rho^{SR^*} \leq 2+\frac{\Psi_{LB}}{L_B}\\
        &\text{where }\Psi_{LB}= \frac{1}{2N}\sum_{i = 1}^{N}w_{i}\left(\frac{\sum_{S\in \mathcal{I}}\Delta(S)^2\mu_{S}^{LB}}{\sum_{S\in \mathcal{I}}\Delta(S)\mu_{S}^{LB}}+1\right).
    \end{split}
\end{equation}
\end{theorem}
\begin{proof}
See Appendix \ref{proof:thm4.2}.
\end{proof}
The stationary randomized policy $\pi_{LB}$ that solves the lower bound in \eqref{eq:LB_opt} assigns low scheduling probabilities $\mu_S^{LB}$ to sets with large values of $\Delta(S)$, which ensures that ${\Psi_{LB}}/{L_B}$ is a constant. In the extreme scenario of high interference where $\Delta(S) \rightarrow \infty$ for some sets $S$, the stationary randomized policy $R$ would assign corresponding scheduling probabilities $\mu_{S}^{LB}\rightarrow0$, ensuring that the optimality gap does not grow unbounded. Therefore, the optimality gap of the stationary randomized policy is bounded by a constant factor, independent of the network size $N$.

Theorem \ref{thm:opt_ratio} captures the impact of multi-AP interference on the optimality of the stationary randomized policies. In the absence of this interference, i.e., when $\Delta(S)=1$ for all $S \in \mathcal{I}$, the optimality ratio of the stationary randomized policies is upper-bounded by 2 plus a small constant, which recovers the known bounds in \cite{throughput} and \cite{aoi1} for single-AP settings.

While the optimal stationary randomized policy can be computed entirely offline and incurs little overhead during online deployment, its computation is hindered by scalability limitations. The number of decision variables and associated constraints scales as $\mathcal{O}(|\mathcal{N}_k|^K)$. As a result, solving \eqref{eq:J_R_opt} becomes computationally intractable for large-scale systems due to both computation and memory constraints. Moreover, the stationary randomized policy needs to recompute $\left\{\mu_{{S}}^*|S\in\mathcal{I}\right\}$ every time there is a change in the transmission times, making it non-robust to real-time deployment and mobile users.

\section{Lyapunov Policy}\label{section:mw}

In AoI optimization literature, authors typically establish an upper bound on the optimal stationary randomized policy, and then introduce either a Lyapunov drift based Max-Weight policy \cite{aoi1,correlated, LargeAndSmall} or a Whittle's index policy \cite{aoi1}. In these works, Lyapunov drift analysis is performed for equal slot or frame lengths. However, in our setting, we have frame lengths that vary over time and also depend on the scheduling actions. This makes standard Lyapunov analysis hard to apply directly. Furthermore, since the evolution of AoI across different users is coupled combinatorially, using a Whittle index approach is also not feasible - as it relies on decoupling across arms.

To overcome this, we propose a new class of Lyapunov drift based Max-Weight policies that can be utilized for varying frame lengths. A policy with similar structure has been proposed by Neely in \cite{neely}. However, the analysis in \cite{neely} is only applicable to renewal processes and throughput optimality. To the best of our knowledge, the analysis of our proposed Lyapunov policy is fundamentally novel, wherein we optimize AoI in a non-renewal system with variable frame lengths to establish bounds.

Our Max-Weight policy selects, at the beginning of each frame, the subset of users that minimize the ratio of the Lyapunov drift over the chosen frame length. In other words, the Max-Weight policy minimizes average drift across time slots. We define a quadratic Lyapunov function 
\begin{equation}
    L(t) = \sum_{i=1}^{N} w_iA_i^2(t).
\end{equation}
Recall that frame $b$ begins at time-slot $t_b$. We define the \textit{drift ratio} at the beginning of frame b as follows:

\begin{equation}
    \phi(S, \mathbf{A}(t_b)) = \mathbb{E} \left[ \frac{L(t_{b+1}) - L(t_b)}{t_{b+1} - t_b} \mathrel{\bigg|} \mathbf{A}(t_b), S \right],
    \label{eq:expected_drift_rate}
\end{equation}
where $\Delta(S) = t_{b+1} - t_b$. The Max-Weight policy, at the beginning of each frame, selects a subset of users that minimizes the ratio of Lyapunov drift across frames to the chosen frame length.
\begin{equation}
\label{eq:max_weight}
    \pi^{MW}(t_b) = \underset{S \in \mathcal{I}}{\arg\min} \,\,  \phi(S, \mathbf{A}(t_b))
\end{equation}
This minimization procedure can be simplified to the maximization below. The derivation details are provided in Appendix \ref{appendix:mw_derivation}.
\begin{equation}
\label{eqn:mw_policy}
\begin{split}
    \pi^{MW}(t_b) &= \underset{S \in \mathcal{I}}{\arg\max} \Bigg\{ \sum_{i \in S_b} w_i \left( 2 A_i(t_b) + \frac{A_i^2(t_b)}{\Delta(S)} \right) \\
    &\quad - \sum_{i=1}^N w_i \Delta(S) \Bigg\}.
\end{split}
\end{equation}

Here $\mathbf{A}(t_b)$ represents the age vector at the beginning of frame b. Notice that the decision taken at the beginning of each frame b, also impacts the frame duration. Further, setting $\Delta(S)=1, \forall S$ we recover the standard Max-Weight policy proposed in \cite{aoi1}.

The following theorem shows that the max-weight policy defined in (\ref{eqn:mw_policy}) is constant-factor optimal with respect to the lower bound defined in Section 3.
\begin{theorem}
\label{thm:mw_optimality_ratio}
Let $J^{MW}$ denote the long-term weighted average AoI under the Max-Weight policy. The optimality ratio $\rho^{MW} = \mathbb{E}[J^{MW}] / L_B$ satisfies:
\begin{equation}
\rho^{MW} \le \bigg(2 + \frac{\Psi_{LB}}{L_B}\bigg)\bigg(\sqrt{2 \bar{\Delta}_{SR}} + \frac{\Psi_{MW}}{L_B}\bigg),
\end{equation}
    where $\Psi_{LB}$ is defined in (\ref{eq:opt_sr}), $\bar{\Delta}_{SR} = \sum_S \Delta(S) \mu^*_S $ is the average frame length under the optimal stationary randomized policy, and $\Psi_{MW} = \frac{\sum_{i=1}^N w_i}{2N} \left( \frac{\bar{\mathbb{M}}[\Delta^2(\pi^{MW}(t_b))]}{\bar{\mathbb{M}}[\Delta(\pi^{MW}(t_b))]} - 1 \right)$.
\end{theorem}

\begin{proof}
See Appendix \ref{proof:thm_mwopt}.
\end{proof}
While we derive a constant factor optimality above that is larger than the optimal stationary randomized policy, this is an artifact of our Lyapunov analysis. We will show through simulations in Section \ref{sec:simulations} that the Max-Weight policy significantly outperforms the optimal stationary randomized policy in practice.

\subsection{Approximating Max-Weight via Submodularity} \label{section:amw}

Finding the optimal set $\pi^{MW}(t_b)$ in (\ref{eqn:mw_policy}) involves a combinatorial search over $\prod_{k=1}^K (|\mathcal{N}_k| + 1)$ possible combinations. This search space grows exponentially with the number of APs $K$, making real-time execution on an SDN controller computationally expensive. However, this optimization can be performed efficiently if the problem exhibits specific structural properties. In the following discussion, we characterize the scheduling constraints as a partition matroid and assume that the set function being maximized in (\ref{eqn:mw_policy}) is submodular. This structure allows us to move away from exhaustive search in favor of efficient locally greedy approximations.

A set function $f(S)$ defined for all subsets $S$ of a finite ground set $\mathcal{N}$ is said to be submodular if it satisfies \cite{bertsimas}:
\begin{equation}
\label{eq:submodular}
    f(S \cup \{i, j\}) - f(S \cup \{i\}) \leq f(S \cup \{j\}) - f(S)
\end{equation}
Formally, this condition reflects the principle of diminishing marginal returns, implying that the `value' of activating a new node is higher when the existing set of active nodes is small.

In our scheduling context, the objective is to maximize the negative drift ratio $-\phi(S)$ at each frame $b$. Derived from the term minimized in \eqref{eq:max_weight}, this expression simplifies to:
\begin{equation}
\label{eq:aggregate_index}
\begin{split}
    -\phi(S) &= \sum_{i \in S} w_i \left( 2 A_i(t_b) + \frac{A_i^2(t_b)}{\Delta(S)} \right) \\
    &\quad - \sum_{k=1}^N w_k \big(\Delta(S) + 2A_i(t_b)\big).
\end{split}
\end{equation}
The decision set $S$ must respect the physical constraints of the network, which are equivalent to describing a partition matroid $\mathcal{I}$. Let the ground set $\mathcal{N}$ be partitioned into disjoint sets $\mathcal{N}_1, \dots, \mathcal{N}_K$, where $\mathcal{N}_k$ represents nodes associated with AP $k$. The feasible set of scheduling actions is $\mathcal{I} = \{ S \subseteq \mathcal{N} : |S \cap \mathcal{N}_k| \leq 1, \forall k \}$.

The submodular behavior of the negative drift ratio $-\phi(S)$ is motivated by the physical nature of interference. Submodularity represents diminishing returns. Adding a user to a small transmission set provides more value than adding one to a large, already crowded set. As more nodes transmit, interference grows and the frame duration $\Delta(S)$ increases. This reduces the marginal benefit of each additional update. Based on this intuition, we formalize the structural nature of our objective function:
\begin{assumption}
\label{asm:submodular_logic}
The negative drift ratio $-\phi(S)$ is a possibly non-monotonic submodular set function.
\end{assumption}

Under this assumption, we can use an approximate local search procedure from \cite{submodular} to find a scheduling decision in polynomial time.

\begin{algorithm}[h]
\caption{Iterative Local Search for Multi-AP Scheduling}
\label{alg:combined_scheduling}
\begin{algorithmic}[1]
\STATE \textbf{Input:} Set of users $\mathcal{N}$, interference constraint set $\mathcal{I}$, parameter $\epsilon$, and objective $-\phi(\cdot)$.
\STATE Initialize $V_1 \leftarrow \mathcal{N}$ and candidate collection ${S} \leftarrow \emptyset$.
\FOR{$m = 1, 2$}
    \STATE Find $v \leftarrow \arg \max \{ -\phi(\{u\}) \mid u \in V_m \}$ and set $S_m \leftarrow \{v\}$.
    
    \WHILE{a local improvement exists}
        \STATE \textbf{Delete:} If $\exists i \in S_m$ s.t. $-\phi(S_m \setminus \{i\}) \geq (1 + \frac{\epsilon}{N^4}) (-\phi(S_m))$:
        \STATE \quad $S_m \leftarrow S_m \setminus \{i\}$.
        
        \STATE \textbf{Exchange:} If $\exists j \in V_m \setminus S_m$ and $i \in S_m \cup \{\emptyset\}$ s.t. 
        \STATE \quad $(S_m \setminus \{i\}) \cup \{j\} \in \mathcal{I}$ and 
        \STATE \quad $-\phi((S_m \setminus \{i\}) \cup \{j\}) \geq (1 + \frac{\epsilon}{N^4}) (-\phi(S_m))$:
        \STATE \quad $S_m \leftarrow (S_m \setminus \{i\}) \cup \{j\}$.
    \ENDWHILE
    
    \STATE ${S} \leftarrow {S} \cup \{S_m\}$ and update ground set: $V_{m+1} \leftarrow V_m \setminus S_m$.
\ENDFOR
\RETURN $\pi^{AMW}(t_b) = S$.
\end{algorithmic}
\end{algorithm}

The algorithm executes two distinct passes, initializing each iteration with the highest-valued individual node from the current available ground set. Within each pass, the solution is refined through \textbf{Delete}, \textbf{Add}, and \textbf{Exchange} operations. The Delete operation addresses the possibly non-monotone nature of $-\phi(S)$ by removing nodes whose interference penalty outweighs their freshness benefit. The Add operation incorporates new nodes that improve the objective while respecting the matroid constraint. The Exchange operation swaps current nodes for superior candidates while maintaining the one-node-per-AP constraint. This iterative search is bounded by $\mathcal{O}(\frac{1}{\epsilon} N^4 \log N)$ local operations, ensuring a polynomial runtime. Consequently, the AMW policy remains computationally efficient for real-time scheduling, avoiding the exponential complexity of exact combinatorial search.

\begin{lemma}
\label{lemma:frame_based_optimality}
At any scheduling frame $b$ with the same AoI values, the approximate set $\pi^{AMW}(t_b)$ identified by the iterative local search algorithm satisfies:
\begin{equation}
\begin{split}
    &\frac{L^{AMW}(t_{b+1}) - L^{AMW}(t_{b})}{\Delta(\pi^{AMW}(t_b))} \\
    &\quad \leq (4+\epsilon) \frac{L^{MW}(t_{b+1}) - L^{MW}(t_{b})}{\Delta(\pi^{MW}(t_b))},
\end{split}
\end{equation}
where $\pi^{MW}(t_b)$ is the Max-Weight subset defined in \eqref{eqn:mw_policy}.
\end{lemma}
\begin{proof}
\label{proof: approx_algo} The proof is an application of results from \cite{submodular}. 
\end{proof}

Further, the per-frame approximation guarantees immediately translate to AoI performance guarantee for the entire time horizon.
\begin{corollary}
    \label{corollary: approx_optimality} Let $J^{AMW}$ denote the long-term weighted average AoI under the max-weight policy solved by the approximation algorithm. The optimality ratio $\rho^{AMW} = \mathbb{E}[J^{AMW}] / L_B$ satisfies:
\begin{equation}
\rho^{AMW} \leq \bigg(2 + \frac{\Psi_{LB}}{L_B}\bigg)\bigg((4+\epsilon) \cdot\sqrt{2 \bar{\Delta}_{SR}} + \frac{\Psi_{MW}}{L_B}\bigg).
\label{eq:corollary}
\end{equation} 
\end{corollary}

\begin{proof}
    See Appendix \ref{proof: approx_algo}
\end{proof}

Note that while submodularity is an intuitive property for our system to hold, proving it is tricky even for the simple SINR-based model discussed in Section~\ref{sec:model}. We will show via simulations in Section \ref{sec:simulations} that the Approximate Max-Weight (AMW) policy performs almost as well as the Max-Weight policy that executes full combinatorial search in each frame, demonstrating the practical applicability of our approach and providing numerical evidence for our assumption.

\section{Simulation Results}
\label{sec:simulations}

In this section, we evaluate the performance of scheduling policies in terms of their long-term weighted average AoI. Specifically, we compare the following policies: the multi-AP optimal stationary randomized policy proposed in Section \ref{section:srp} (SRP), the multi-AP Max-Weight policy proposed in Section \ref{section:mw} (MW), the multi-AP Approximate Max-Weight policy proposed in Section \ref{section:amw} (AMW), and the single-AP Max-Weight policy proposed in \cite{aoi1}, which we execute in a distributed manner across APs. 

We use the single-AP Max-Weight policy as our baseline, as it has been shown to achieve near-optimal performance in single-hop AoI settings both theoretically and in practice \cite{uav1}. We note that the stationary randomized policy in \cite{interference} is derived under a simplified binary interference model with updates that finish in a single time-slot. Thus, it does not capture concurrent transmissions and non-uniform frame lengths possible in our model and cannot be directly utilized for our setting. 

We simulate a multi-AP wireless network consisting of $N$ users and $K$ APs. Each user is assigned to the geographically closest AP. All users transmit updates of size 100 KB at a transmit power level of 10 mW, over a 22 MHz channel (which is standard in 2.4 GHz WiFi). Wireless links are modeled using an SINR-based model with unit fading gain, and a noise density of -174 dBm/Hz. Under this setting, when the path-loss exponent is 2.5, a user located 5 m from its associated AP requires about 1.2 ms to complete a transmission. We run the simulation over a duration of 1 s and report the average AoI in milliseconds. Since there is randomness due to user locations, we repeat each experiment with 10 different random seeds and report the average across 10 runs along with error bars. 

The multi-AP policies make scheduling decisions synchronously across the network, whereas the single-AP Max-Weight policy operates locally and asynchronously at each AP. For all policies, the AoI of each node is reset immediately upon packet completion, independent of other ongoing transmissions, \textit{thereby capturing the asynchronous nature of packet completions in practical systems.}
\begin{figure}[ht]
    \centering
    \includegraphics[width=0.4\textwidth]{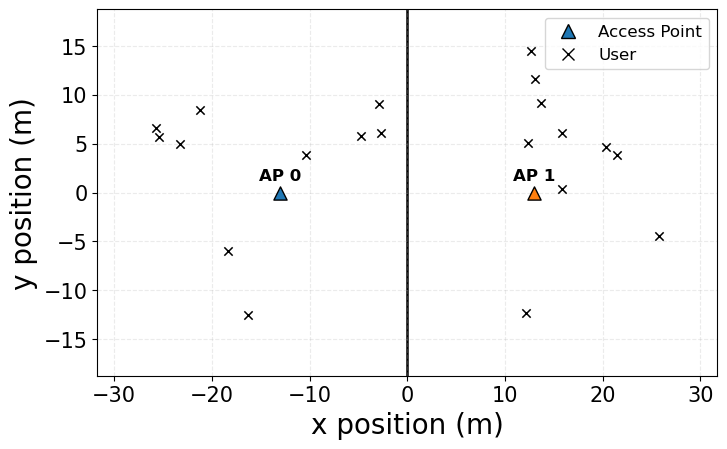}
    \caption{Example of the multi-AP WLAN deployment used in the simulations with $K = 2$ APs and $N=20$ users. The interference between the 2 APs is modeled using the spectrum overlap factor $\eta \in (0,1)$. }
    \label{fig:2_AP_hex_layout}
\end{figure}
For our first set of experiments, we focus on a network with two APs and plot performance for two cases - when there are 10 users in the system and when there are 20 users in the system (see Figure \ref{fig:2_AP_hex_layout}). In each case, the users are equally split between the 2 APs. The two APs are separated by a distance of $15 m$, and users are uniformly distributed within a $15 m$ radius of their associated AP. The path-loss exponent is set to 3.5 for this set of experiments. 

We are interested in measuring how the performance of different policies evolves as a function of the spectrum overlap $\eta$ between the 2 APs. For zero overlap, we expect that distributed single-AP Max-Weight should perform well. However, as the overlap increases, the amount of interference between the two sets of users should also increase, leading to performance gains for coordinated scheduling policies. Figure \ref{fig:aoi_overlap_sweep} confirms this observation. When $\eta = 0$, the distributed single-AP Max-Weight policy achieves the same performance as the coordinated multi-AP Max-Weight policy and outperforms the multi-AP optimal SRP. As interference ($\eta$) increases, performance of both the multi-AP SRP and the multi-AP Max-Weight policy remains flat, with the multi-AP max-weight policy staying close to the lower bound, while the average AoI of single-AP Max Weight grows worse. At $\eta=0.5$, our MW policy outperforms the single-AP policy by about $\mathbf{50\%}$.

Figure \ref{fig:aoi_overlap_sweep} also shows that the closed-form AoI of the multi-AP optimal SRP derived under the synchronous packet completion assumption closely matches the numerical simulation results where packet completions are asynchronous. This supports our claim that the AoI analysis and policies derived under the synchronous packet completions assumption still generalize well under asynchronous packet completions.
\begin{figure}[ht]
    \centering
    \includegraphics[width=0.48\textwidth]{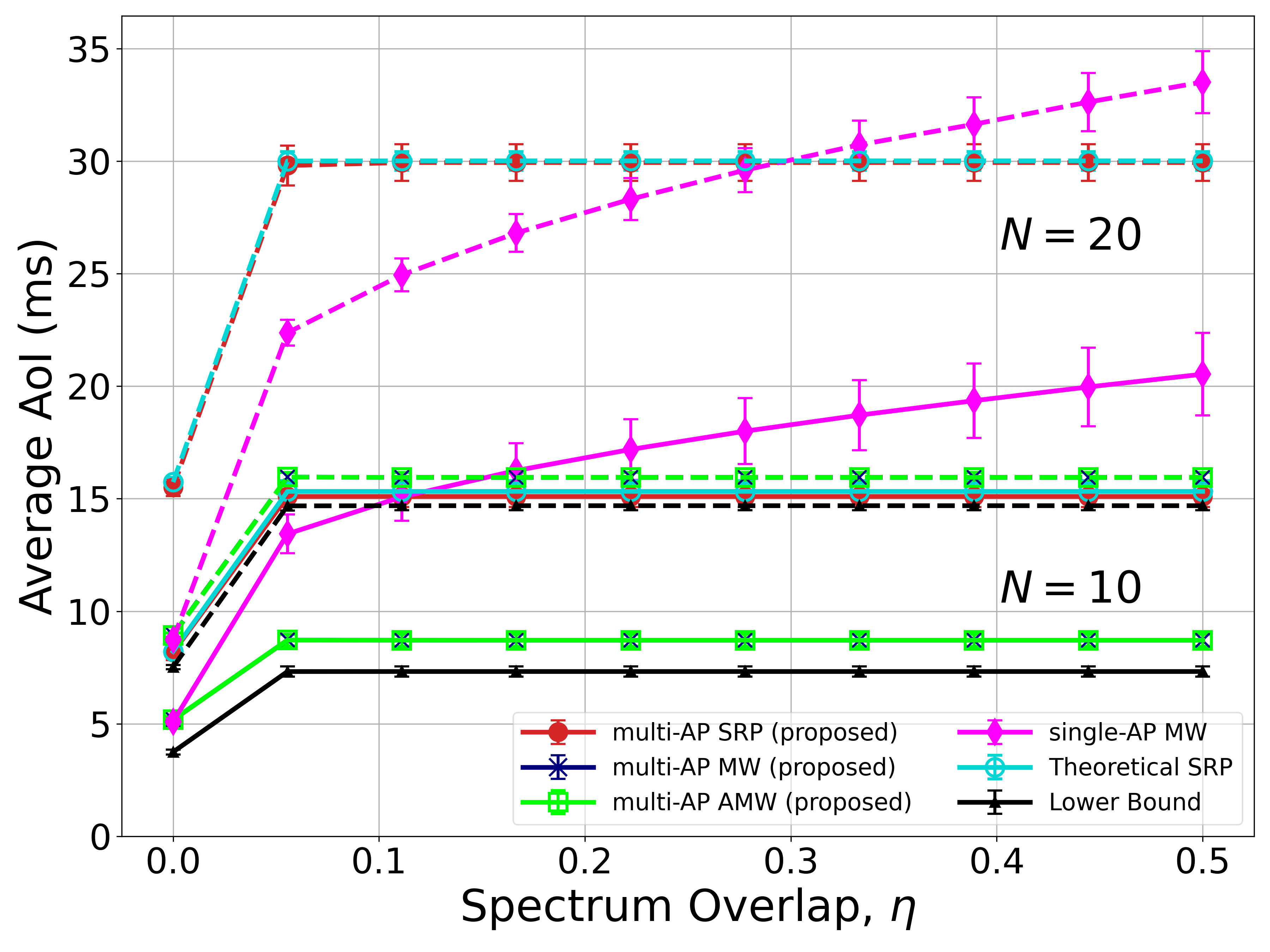}
    \caption{Plot of the average AoI for the 2-AP network as a function of the spectrum overlap factor $\eta$, evaluated at 10 uniformly spaced values in $[0,0.5]$. Solid and dashed lines correspond to $N=10$ and $N=20$, respectively.}
    \label{fig:aoi_overlap_sweep}
\end{figure}

Next, we consider settings with more than 2 APs. Figure \ref{fig:layout_fig} describes the layout for a 9 AP wireless network deployed over a hexagonal grid, operating in the 2.4 GHz WiFi band. Each AP is located at the center of a hexagonal cell of radius $7m$ and serves $5$ users uniformly distributed within its coverage area. Thus, $K=9$ and $N=45$ for the network layout drawn in Figure \ref{fig:layout_fig}. A 3-channel reuse pattern of the 3 orthogonal channels $\{1,6,11\}$ in the 2.4\,GHz Wi-Fi band is assigned to the APs such that adjacent cells operate on orthogonal channels, which is standard for 2.4GHz WiFi deployment in the US. As a result, there is no interference from neighboring APs, and co-channel interference arises from spatially separated cells that reuse the same channel. While we have chosen an orthogonal channel assignment scheme, our framework can easily handle non-orthogonal channel assignments as illustrated by the 2 AP experiments. 
\begin{figure}[ht]
    \centering
    \includegraphics[width=0.48\textwidth]{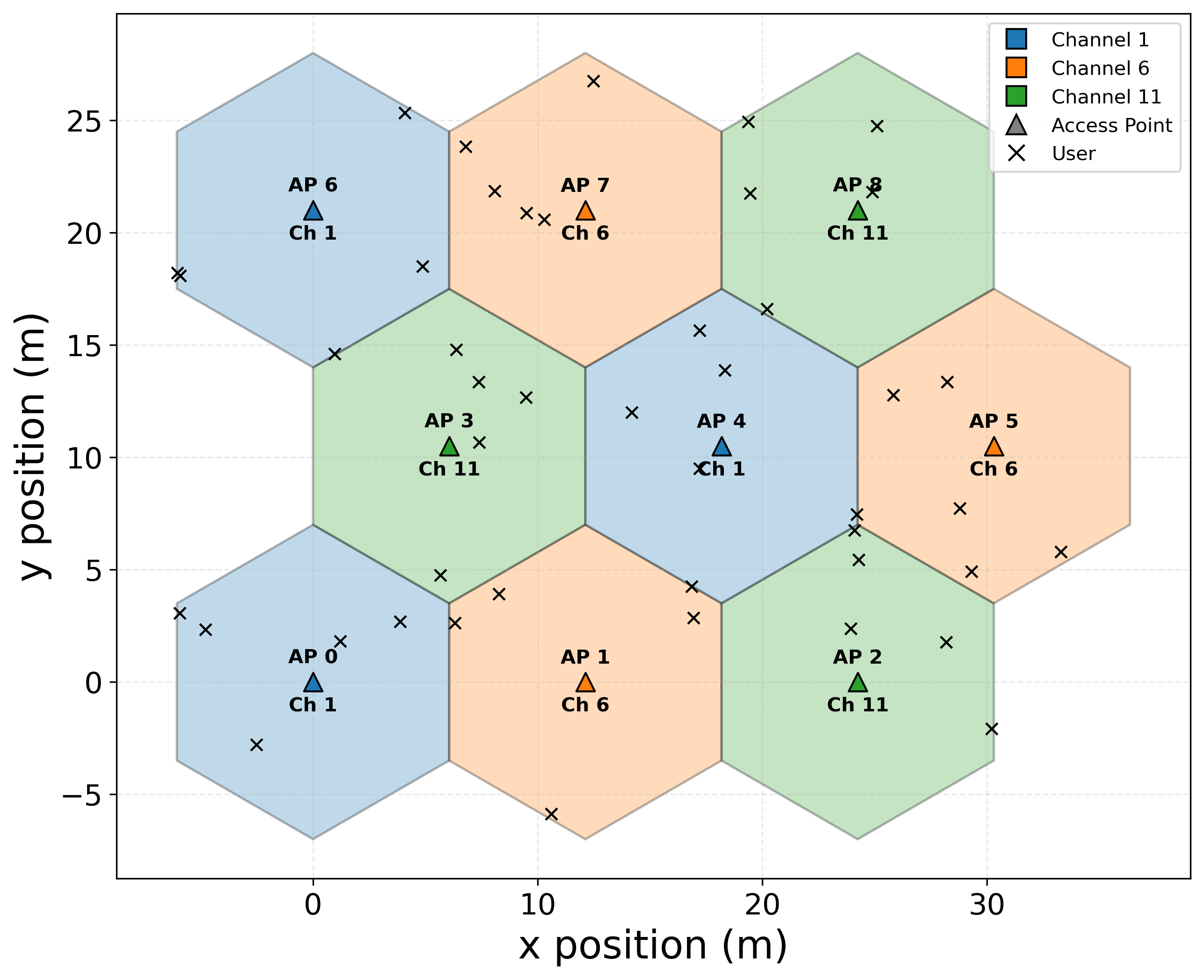}
    \caption{Example of the multi-AP WLAN deployment used in the simulations with $K=9$ APs and $N=45$ users. APs are arranged on a regular hexagonal grid with cell radius $R=7\,\mathrm{m}$, and 5 users are distributed uniformly within each cell. A 3-channel reuse pattern $\{1,6,11\}$ in the 2.4\,GHz Wi-Fi band is assigned to the APs.}
    \label{fig:layout_fig}
\end{figure}

We evaluate the scaling performance of the proposed policies in Figure \ref{fig:aoi_compare_fig}. Our centralized policies show a clear advantage over the single-AP MW baseline, with the multi-AP MW policy nearly achieving the lower bound. The AMW variant closely matches the performance of the multi-AP MW policy but offers a significant computational speedup. Consider the $K=9$ case. The average AoI achieved by the uncoordinated baseline policy is 27ms, while the AMW policy achieves an AoI of 14ms. This is an AoI reduction of nearly $\mathbf{48\%}$ compared to the uncoordinated distributed approach. 

To verify our theoretical results, we evaluate the performance guarantees of our proposed policies in Table~\ref{tab:theoretical_bounds}. The guarantee for AMW can also be constructed from the $\rho^{MW}$ term. Crucially, the constants decrease with network size, suggesting that our theoretical guarantees remain well-behaved even in larger WLANs. Furthermore, as seen in Figure~\ref{fig:aoi_compare_fig}, the empirical performance ratio between the lower bound and the MW policy is $\approx$ 1.1. This is significantly tighter than the theoretical bound of $\approx$ 4.8 shown in Table~\ref{tab:theoretical_bounds}. This discrepancy is an artifact of the worst-case nature of multi-slot Lyapunov drift analysis and is commonly observed in AoI optimization literature \cite{aoi1, throughput, LargeAndSmall}.

\begin{table}[ht]
\centering
\caption{{Evaluation of theoretical bounds $\rho^{SR^*}$ and $\rho^{MW}$ under uniform and random weights.}}
\label{tab:theoretical_bounds}
\resizebox{\columnwidth}{!}{%
\begin{tabular}{c|cc|cc}
\hline
& \multicolumn{2}{c|}{$\rho^{SR^*}$} & \multicolumn{2}{c}{$\rho^{MW}$} \\
\# APs & $w_i=1$ & $w_i\in[1,10]$ & $w_i=1$ & $w_i\in[1,10]$ \\
\hline
3 & 2.201 & 2.193 & 5.153 & 5.124 \\
5 & 2.112 & 2.103 & 4.772 & 4.751 \\
7 & 2.074 & 2.072 & 4.627 & 4.621 \\
9 & 2.067 & 2.102 & 4.583 & 4.618 \\
\hline
\end{tabular}%
}
\end{table}

Note that we have only reported the average AoI for up to $9$ APs and $45$ users for the lower bound and the optimal SRP policy, since we hit memory and computation constraints for larger network sizes. This also highlights the exponential computational complexity associated with the SRP policy as the number of APs increase. The AMW policy, however can easily be run for larger networks due to its polynomial complexity. 
\begin{figure}[ht]
    \centering
    \includegraphics[width=0.48\textwidth]{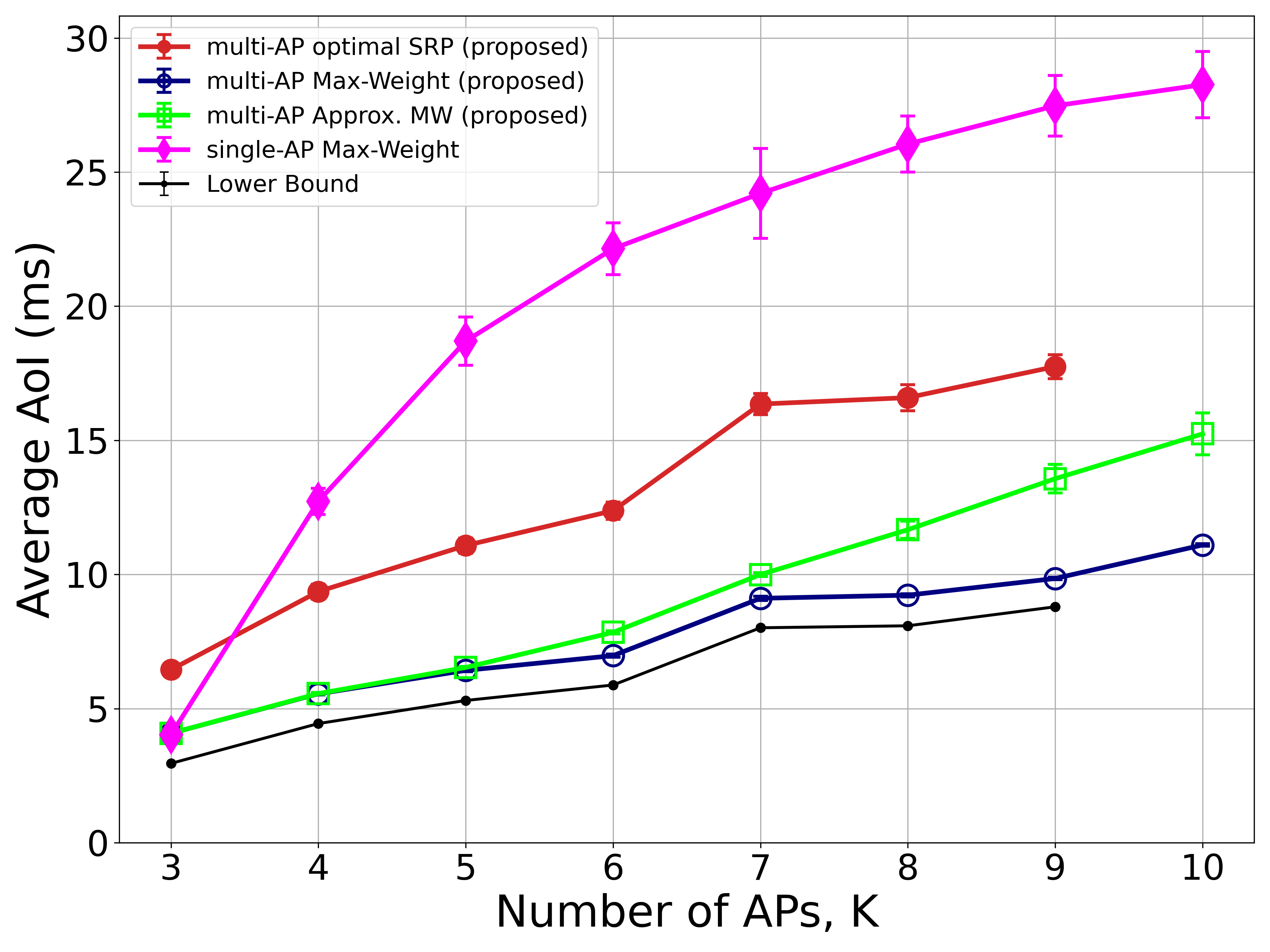}
    \caption{Plot of the average AoI as a function of the number of access points, $K$.}
    \label{fig:aoi_compare_fig}
\end{figure}

The computational overhead for our three proposed policies is reported in Tables \ref{tab:srp_offline_time} and \ref{tab:decision_time_comparison}. The optimal SRP requires extensive offline computation, as seen in Table \ref{tab:srp_offline_time}, but it can be executed in real-time with negligible overhead once the optimal probabilities are determined. More importantly, it needs frequent recomputation every time the locations of users or channel assignments or the wireless environment changes. This makes the policy non-robust and impractical for deployment due to the exponentially growing offline computation times and memory requirements.

\begin{table}[ht]
\centering
\caption{Offline computation time for the optimal stationary randomized policy with an increasing number of nodes ($N$) and APs ($K$).}
\label{tab:srp_offline_time}
\begin{tabular}{cc}
\hline
\textbf{$(N,K)$} & \textbf{Total Offline Compute Time (s)} \\
\hline
(15,3)  & 0.80 \\ 
(30,6)  & 1.26 \\
(45,9)  & 361.71 \\
\hline
\end{tabular}
\end{table}
For the multi-AP MW and AMW policies, we report the average computation time required to solve the maximization step in each scheduling frame in Table \ref{tab:decision_time_comparison}. The max-weight policy, similar to the SRP, has an exponential increase in decision times as N and K increase. The empirical decision times confirm our theoretical analysis of exponential increase. The AMW policy, on the other hand, is very efficient, and as shown in Table \ref{tab:decision_time_comparison}, the overhead increases only marginally with N and K. Coupled with the fact that it also performs close to the MW policy in terms of AoI, AMW is the best candidate for implementation in real-time without compromising on performance. Further, in a large scale network, we can use multiple controllers for scheduling, where each controller only coordinates in smaller groups of 5-10 APs, increasing scalability.

Finally, one aspect that we have not accounted for in our simulations is the added latency overhead due to centralized SDN coordination and signaling. Given the significant performance gains in freshness and the relative low cost of executing the AMW policy, we believe that multi-AP coordination can be a very useful approach for dense WLANs, especially for large update sizes, even with the added SDN control overheads.

\begin{table}[ht]
\centering
\caption{Comparison of average decision time per frame for MW and AMW policies with an increasing number of nodes ($N$) and APs ($K$)}
\label{tab:decision_time_comparison}
\begin{tabular}{ccc}
\hline
\textbf{$(N,K)$} & \textbf{MW Policy (ms)} & \textbf{AMW Policy (ms)} \\
\hline
(15,3)  & 0.331     & 0.070 \\ 
(30,6)  & 8.153     & 0.076 \\
(45,9)  & 1353.300  & 0.085 \\
\hline
\end{tabular}
\end{table}

\section{Conclusion}

This work addresses the challenge of optimizing information freshness in dense, multi-AP WLANs where inter-cell interference couples scheduling decisions. By modeling the system through a frame-based abstraction with action-dependent durations, we derived a fundamental lower bound on the achievable weighted sum AoI. We developed two classes of coordinated scheduling policies: a stationary randomized approach and an online Max-Weight policy. Notably, our ratio-based Lyapunov drift analysis provides a novel framework for establishing performance guarantees in non-renewal systems with variable frame lengths.

To ensure scalability for large-scale enterprise networks, we utilized a submodularity assumption to design an efficient approximation for the Max Weight policy. Simulation results confirm that our centralized policies effectively reduce AoI, significantly outperforming uncoordinated distributed baselines. Interesting extensions include settings with mobile users, online estimation of transmission times, and practical implementation on an SDN controller.

\bibliographystyle{IEEEtran}
\bibliography{sample-base}

\appendices
\section{Proof of Lemma~\ref{lemma:J_pi}}\label{proof:lemma_1}
The time horizon can be expressed as:
\begin{equation} \label{eq:4}
    T = \sum_{m = 1}^{D_{i}(T)}(W_{i}[m]+S_{i}[m])+R_{i}.
\end{equation}
The time-average AoI associated with source $i$ can be expressed as:
\begin{equation} \label{eq:5}
\begin{split}
    \frac{1}{T}\sum_{t=1}^{T}A_{i}(t) = &\frac{1}{T}\left[\sum_{m =1}^{D_{i}(T)}\frac{(W_{i}[m]+S_{i}[m])^2}{2}\right.\\
    &+\left.(S_{i}[m-1]+\frac{1}{2})(W_{i}[m]+S_{i}[m])\right]\\
    &+\frac{1}{T}\left[(S_{i}[D_{i}]+\frac{1}{2})R_{i}+R_{i}^2\right].
\end{split}
\end{equation}
Combining (\ref{eq:4}), (\ref{eq:6}), and (\ref{eq:7}) yields
\begin{equation} \label{eq:8}
\begin{split}
    \frac{T}{D_{i}(T)} &= \frac{1}{D_{i}(T)}\left(\sum_{m = 1}^{D_{i}(T)}(W_{i}[m]+S_{i}[m])+R_{i}\right)\\
    &=\frac{1}{D_{i}(T)}\sum_{m = 1}^{D_{i}(T)}W_{i}[m]+\frac{1}{D_{i}(T)}\sum_{m = 1}^{D_{i}(T)}S_{i}[m] \\&+\frac{R_{i}}{D_{i}(T)}\\
    &=\mathbb{\bar{M}}[W_{i}[m]]+\mathbb{\bar{M}}[S_{i}[m]]+\frac{R_{i}}{D_{i}(T)}.
\end{split}
\end{equation}
Combining (\ref{eq:5}) and (\ref{eq:8}) yields
\begin{equation}\label{eq:9}
    \begin{split}
        \frac{1}{T}\sum_{t=1}^{T}A_{i}(t) = &\frac{D_{i}(T)}{T}\mathbb{\bar{M}}\left[\frac{(W_{i}[m]+S_{i}[m])^2}{2}\right.\\
        &\left.+(S_{i}[m-1]+\frac{1}{2})(W_{i}[m]+S_{i}[m])\right]\\
        &+\frac{1}{T}\left[(S_{i}[D_{i_k}]+\frac{1}{2})R_{i}+R_{i}^2\right]\\
        &=\frac{\mathbb{\bar{M}}\left[\frac{(W_{i}[m]+S_{i}[m])^2}{2}\right]}{\mathbb{\bar{M}}[W_{i}[m]]+\mathbb{\bar{M}}[S_{i}[m]]+\frac{R_{i_k}}{D_{i}(T)}}\\
        &+\frac{\mathbb{\bar{M}}[(S_{i}[m-1]+\frac{1}{2})(W_{i}[m]+S_{i}[m])]}{\mathbb{\bar{M}}[W_{i}[m]]+\mathbb{\bar{M}}[S_{i}[m]]+\frac{R_{i}}{D_{i}(T)}}\\
        &+\frac{1}{T}\left[(S_{i}[D_{i}]+\frac{1}{2})R_{i}+R_{i}^2\right].\\
    \end{split}
\end{equation}
Taking the limit of (\ref{eq:9}) as $T\rightarrow\infty$ and assuming $R_{i}<\infty$, we have $D_{i}(T) \rightarrow \infty$, $\frac{R_{i}}{T} \rightarrow 0$,$\frac{R_{i}^2}{T} \rightarrow 0$. Applying the limits to (\ref{eq:9}), we get
\begin{equation}
    \begin{split}
        \lim_{T\rightarrow\infty}\frac{1}{T}\sum_{t=1}^{T}A_{i}(t) =& \lim_{T\rightarrow \infty}\left[\frac{\mathbb{\bar{M}}[(W_{i}[m]+S_{i}[m])^2]}{2\mathbb{\bar{M}}[W_{i}[m]+S_{i}[m]]}\right.\\
        &\left.+\frac{\mathbb{\bar{M}}[S_{i}[m-1](W_{i}[m]+S_{i}[m])]}{\mathbb{\bar{M}}[W_{i}[m]+S_{i}[m]]}+\frac{1}{2}\right].
    \end{split}
\end{equation}
Taking the weighted sum average across all $N$ nodes with respect to weights $w_{i}$ yields (\ref{eq:11}).

\section{Proof of Theorem~\ref{thm:srp_ewsaoi}}\label{proof:thm4.1}
For Stationary Randomized policies, since the scheduling decision is made independent of history information, the waiting time $W_{i}[m]$ and service time $S_{i}[m]$ are independent across packet updates. $W_{i}[m]$ and $S_{i}[m]$ are also independent of each other within the delivery of an update $m$. Taking the expectation of (\ref{eq:11}) and omitting the update index yields the long-term weighted average AoI of any Stationary Randomized policy $R\in\Pi^{sr}$:
\begin{equation}\label{eq:13}
    \begin{split}
        \mathbb{E}[J^R] = & \frac{1}{N}\sum_{i = 1}^{N}w_{i}\left[\frac{\mathbb{E}[W_{i}^2]+2\mathbb{E}[W_{i}]\mathbb{E}[S_{i}]+\mathbb{E}[S_{i}^2]}{2\mathbb{E}[W_{i}]+2\mathbb{E}[S_{i}]}\right.\\
        &\left.+\frac{\mathbb{E}[S_{i}](\mathbb{E}[W_{i}]+\mathbb{E}[S_{i}])}{\mathbb{E}[W_{i}]+\mathbb{E}[S_{i}]}+\frac{1}{2}\right].
    \end{split}
\end{equation}
The expected service time of user $i$ is the expected transmission time for the scheduled set over all possible sets containing user $i$. The set-wide transmission time $\Delta(S)$ is the service time when the set $S$ is selected. The expected service time is
\begin{equation}\label{eq:14}
    \mathbb{E}[S_{i}] = \frac{\sum_{ \{S\in \mathcal{I}: i\in S \}}\Delta(S)\mu_S}{\sum_{ \{S\in \mathcal{I}: i\in S \}}\mu_S}
\end{equation}
Similarly, the second moment of the service time is
\begin{equation}\label{eq:15}
    \mathbb{E}[S_{i}^2] = \frac{\sum_{ \{S\in \mathcal{I}: i\in S \}}\Delta(S)^2\mu_S}{\sum_{ \{S\in \mathcal{I}: i\in S \}}\mu_S}
\end{equation}

After the current update from user $i$ finishes transmission, the wait time $W_{i}$ is 0 if a set S containing user $i$ is selected for the next update. If another set is selected, then the wait time until the next update from user $i$ will be $\Delta(S)+\mathbb{E}[W_{i}], i\notin S.$ Applying the Law of Total Expectation yields
\begin{equation} \label{eq:E_W}
    \begin{split} 
        \mathbb{E}[W_{i}] = &\sum_{ \{S\in \mathcal{I}: i\in S \}}\mathbb{E}[W_{i}|u_{i}=1]\mu_{S} \\
        &+\sum_{ \{S\in \mathcal{I}: i \notin S \}}\mathbb{E}[W_{i}|u_{i}=0]\mu_{S}\\
        &=\sum_{ \{S\in \mathcal{I}: i\in S \}} 0\cdot\mu_{S}\\
        &+\sum_{ \{S\in \mathcal{I}: i \notin S \}}(\Delta(S)+\mathbb{E}[W_{i}])\mu_{S}.\\
    \end{split}
\end{equation}
Rearranging the equation yields
\begin{equation}\label{eq:17}
\begin{split}
        \mathbb{E}[W_{i}] &= \frac{\sum_{ \{S\in \mathcal{I}: i \notin S \}}\Delta(S)\mu_{S}}{\sum_{ \{S\in \mathcal{I}: i\in S \}}\mu_S}\\
    \end{split}
\end{equation}
Similarly, the second moment of the wait time can be written as
\begin{equation}
    \begin{split}
        \mathbb{E}[W_{i}^2] = &\sum_{ \{S\in \mathcal{I}: i\in S \}}\mathbb{E}[W_{i}^2|u_{i}=1]\mu_{S} \\
        &+\sum_{ \{S\in \mathcal{I}: i \notin S \}}\mathbb{E}[W_{i}^2|u_{i}=0]\mu_{S}\\
        &=\sum_{ \{S\in \mathcal{I}: i\in S \}} 0\cdot\mu_{S}\\
        &+\sum_{ \{S\in \mathcal{I}: i \notin S \}}\mathbb{E}[(\Delta(S)+W_{i})^2]\mu_{S}\\
        &= \sum_{ \{S\in \mathcal{I}: i \notin S \}}\left(\Delta(S)^2+\mathbb{E}[W_{i}^2]+2\Delta(S)\mathbb{E}[W_{i}]\right)\mu_{S}.\\
    \end{split}
\end{equation}
Rearranging the equation and substituting \eqref{eq:14}, \eqref{eq:15}, and \eqref{eq:E_W} yields
\begin{equation}\label{eq:19}
    \begin{split}
        \mathbb{E}[W_{i}^2] = &\frac{\sum_{ \{S\in \mathcal{I}: i \notin S \}}\Delta(S)^2\mu_S+\frac{2(\sum_{ \{S\in \mathcal{I}: i \notin S \}}\Delta(S)\mu(S))^2}{\sum_{ \{S\in \mathcal{I}: i \in S \}}\mu(S)}}{\sum_{ \{S\in \mathcal{I}: i\in S \}}\mu_{S}}.
    \end{split}
\end{equation}

Substituting (\ref{eq:14}), (\ref{eq:15}), (\ref{eq:17}) and (\ref{eq:19}) into (\ref{eq:13}) yields
\begin{equation} \label{eq:J_R}
    \begin{split}
        \mathbb{E}[J^R] &= \frac{1}{N}\sum_{i = 1}^{N}w_{i}\Biggl(\frac{\sum_{S\in \mathcal{I}}\Delta(S)^2\mu_{S}}{2\sum_{S\in \mathcal{I}}\Delta(S)\mu_{S}} \\
        &\quad + \frac{\sum_{S\in \mathcal{I}}\Delta(S)\mu_{S}}{\sum_{ \{S\in \mathcal{I}: i\in S \}}\mu_{S}}-\frac{1}{2}\Biggr).
    \end{split}
\end{equation}

\section{Proof of Theorem~\ref{thm:opt_ratio}}\label{proof:thm4.2}
Combining (\ref{eq:8}) and (\ref{eq:throughput_1}), the long-term throughput of node $i$ under the stationary randomized policy can be expressed as 
\begin{equation}
    q_{i}^R = \frac{1}{\mathbb{E}[W_{i}]+\mathbb{E}[S_{i}]}.
\end{equation}
Substituting (\ref{eq:14}) and (\ref{eq:E_W}) into the expression yields
\begin{equation} \label{eq:q_R_mu}
    q_{i}^R = \frac{\sum_{S\in \mathcal{I},i\in S}\mu_{S}}{\sum_{S\in \mathcal{I}}\Delta(S)\mu_{S}}.
\end{equation}
Substituting $q_{i}^R$ into (\ref{eq:J_R_thm}) yields
\begin{equation}\label{eq:J_R_UB_1}
    \begin{split}
        \mathbb{E}[J^R] &= \frac{1}{N}\sum_{i = 1}^{N}w_{i}\left(\frac{\sum_{S\in \mathcal{I}}\Delta(S)^2\mu_{S}}{2\sum_{S\in \mathcal{I}}\Delta(S)\mu_{S}}+\frac{1}{q_{i}^R}-\frac{1}{2}\right).
    \end{split}
\end{equation}
Let $q_{i}^L$ be the long-term throughput associated with the policy that solves the Lower Bound optimization problem in (\ref{eq:LB_opt}). Let $R\in \Pi^{sr}$ denote the stationary randomized policy defined by the scheduling probabilities $\mu_{S}$ such that $\mu_{S}^R$ is a solution to (\ref{eq:q_R_mu}) when $q_{i}^R = q_{i}^L$. Combining (\ref{eq:LB_opt}) and (\ref{eq:J_R_UB_1}) yields
\begin{equation}\label{eq:JR_LB}
    \mathbb{E}[J^R] = L_B\left[2+\frac{\frac{1}{N}\sum_{i = 1}^{N}w_{i}\left(\frac{\sum_{S\in \mathcal{I}}\Delta(S)^2\mu_{S}^R-1}{2+2\sum_{S\in \mathcal{I}}\Delta(S)\mu_{S}^R}+\frac{1}{2}\right)}{L_B}\right].
\end{equation}
Substituting (\ref{eq:JR_LB}) into $\rho^R = \mathbb{E}[J^R]/L_B$ yields
\begin{equation}
    \begin{split}
        &\rho^R = 2+\frac{\Psi}{L_B}\\
    &\text{where }\Psi = \frac{1}{2N}\sum_{i = 1}^{N}w_{i}\left(\frac{\sum_{S\in \mathcal{I}}\Delta(S)^2\mu_{S}^R-1}{1+\sum_{S\in \mathcal{I}}\Delta(S)\mu_{S}^R}+1\right).
    \end{split}
\end{equation}
The first term in $\Psi$ is upper-bounded by the maximum effective transmission time $\Delta_{\text{max}}$ over all possible subsets of transmitting nodes. In the scenario of arbitrarily high interference, where $\Delta_{\text{max}}\rightarrow \infty$, the stationary randomized policy $R$ would assign scheduling probability $\mu_{S}\rightarrow0$ to any pair $S$ with effective transmission time $\Delta_{\text{max}}$ so that $q_{i}^R$ would not become zero and $L_B$ would not grow to infinity. Therefore, $\Psi$ always stays bounded.

Let $R^*\in\Pi^{sr}$ denote the optimal stationary randomized policy that solves (\ref{eq:J_R_opt}). As $\mathbb{E}[J^{R^*}]\leq \mathbb{E}[J^R]$, we obtain the optimality ratio for the optimal stationary randomized policy:
\begin{equation}
    \rho^{R^*} \leq 2+\frac{\Psi}{L_B}.
\end{equation}

\section{Derivation of the Max-Weight Policy}
\label{appendix:mw_derivation}

To derive the Max-Weight policy, we evaluate the drift of the Lyapunov function defined in (\ref{eq:max_weight}). Substituting AoI evolution rules from (\ref{eq:aoi_evolution}) into the expression for $L(t_{b+1})$, we obtain:
\begin{align}
    L(t_{b+1}) &= \sum_{i \in S(t_b)} w_i \Delta^2 \bigl( S(t_b) \bigr) \nonumber \\
    &\quad + \sum_{i \notin S(t_b)} w_i \Bigl( A_i(t_b) + \Delta \bigl( S(t_b) \bigr) \Bigr)^2 \nonumber \\
    &= \sum_{i \in S(t_b)} w_i \Delta^2 \bigl( S(t_b) \bigr) \nonumber \\
    &\quad + \sum_{i \notin S(t_b)} w_i \biggl( A_i^2(t_b) + 2A_i(t_b)\Delta \bigl( S(t_b) \bigr) \nonumber \\
    &\qquad\qquad\qquad\quad + \Delta^2 \bigl( S(t_b) \bigr) \biggr).
\end{align}
The Lyapunov drift $\Delta L(t_b) = L(t_{b+1}) - L(t_b)$ is then simplified by subtracting the current state $L(t_b) = \sum_{i=1}^N w_i A_i^2(t_b)$:
\begin{align}
    \Delta L(t_b) &= \sum_{i \in S(t_b)} w_i \Delta^2 \Bigl( S(t_b) \Bigr) + \sum_{i \notin S(t_b)} w_i \Delta^2 \Bigl( S(t_b) \Bigr) \nonumber \\
    &\quad + \sum_{i \notin S(t_b)} 2w_i A_i(t_b) \Delta \Bigl( S(t_b) \Bigr) \nonumber \\
    &\quad + \sum_{i \notin S(t_b)} w_i A_i^2(t_b) - \sum_{i=1}^N w_i A_i^2(t_b).
\end{align}
Now,

\begin{align}
    \Delta L(t_b) &= \sum_{i=1}^N w_i \Delta^2 \Bigl( S(t_b) \Bigr) + 2\Delta \Bigl( S(t_b) \Bigr) \sum_{i \notin S(t_b)} w_i A_i(t_b) \nonumber \\
    &\quad - \sum_{i \in S(t_b)} w_i A_i^2(t_b). \label{eq:drift_step}
\end{align}
The policy $\pi^{MW}(t_b)$ in (\ref{eq:max_weight}) minimizes the normalized conditional drift ratio $\mathbb{E}[\Delta L(t_b) \mid \mathbf{A}(t_b)] / \Delta(S(t_b))$. Dividing (\ref{eq:drift_step}) by $\Delta(S(t_b))$ and substituting the identity $\sum_{i \notin S(t_b)} w_i A_i(t_b) = \sum_{i=1}^N w_i A_i(t_b) - \sum_{i \in S(t_b)} w_i A_i(t_b)$, we can ignore the constant term $2\sum_{i=1}^N w_i A_i(t_b)$ for the purpose of optimization. This yields the following minimization problem:
\begin{equation}
    \begin{split}
        \min_{S(t_b) \in \mathcal{I}} \Biggl[ &\sum_{i=1}^N w_i\Delta(S(t_b)) - 2 \sum_{i \in S(t_b)} w_i A_i(t_b) \\
        &- \frac{\sum_{i \in S(t_b)} w_i A_i^2(t_b)}{\Delta(S(t_b))} \Biggr].
    \end{split}
\end{equation}
Finally, negating the objective to convert it to a maximization results in the policy form in (\ref{eq:max_weight}).
\qed

\section{Proof of Theorem \ref{thm:mw_optimality_ratio}}
\label{proof:thm_mwopt}

The network state is tracked across scheduling frames indexed by $t_b$. The age of node $i$ at the start of frame $t_b + 1$ is an evolution of its age in the preceding frame $t_b$ and the duration of the current scheduling interval:
\begin{align}
    A_{i}(t_b+1) &= A_{i}(t_b) + \Delta\big(S(t_b)\big) - A_{i}(t_b)u_{i}\big(S(t_b)\big), \label{eq:evolution} \\
    u_{i}\big(S(t_b)\big) &= \begin{cases} 
        1, & \text{if } i \in S(t_b), \\ 
        0, & \text{else}. 
    \end{cases} \label{eq:indicator}
\end{align}
The total accumulated cost for a specific user $i$ over the duration of a single frame is defined as the sum of their age at each time step within that frame interval:
\begin{align}
    & A_{i}(t_b) + A_{i}(t_b) + 1 + \dots + \Big( A_{i}(t_b) + \Delta\big(S(t_b)\big) - 1 \Big). \nonumber \\
    \Rightarrow & A_{i}(t_b)\Delta\big(S(t_b)\big) + \frac{\Big( \Delta\big(S(t_b)\big) - 1 \Big)\Delta\big(S(t_b)\big)}{2}. \label{eq:one_frame_cost}
\end{align}
By aggregating these costs over a horizon of $B$ frames and normalizing by the total elapsed time, we derive the average network cost expression:
\begin{equation}
\begin{split}
    \text{Avg cost} &= \frac{1}{\sum_{t_b=1}^{B} \Delta\big(S(t_b)\big)} \Bigg( \sum_{b=1}^{B} \sum_{i=1}^{N} w_{i} A_{i}(t_b)\Delta\big(S(t_b)\big) \\
    &\quad + \sum_{i=1}^{N} w_{i} \bigg( \frac{\Delta^{2}\big(S(t_b)\big) - \Delta\big(S(t_b)\big)}{2} \bigg) \Bigg).
\end{split}
\label{eq:avg_network_cost}
\end{equation}
Using (\ref{eq:avg_network_cost}), taking expectation and letting $B \rightarrow \infty $, we get
\begin{equation}
\begin{split}
    J^{MW} &= \lim_{B \to \infty} \frac{1}{N \sum_{b=1}^{B} \mathbb{E} \big[ \Delta\big(S(t_b)\big) \big]} \\
    &\quad \times \sum_{b=1}^{B} \mathbb{E} \Bigg[ \sum_{i=1}^{N} w_i \bigg( A_i(t_b) \Delta\big(S(t_b)\big) \\
    &\qquad\qquad\qquad\quad + \frac{\Delta^2\big(S(t_b)\big) - \Delta\big(S(t_b)\big)}{2} \bigg) \Bigg],
\end{split}
\label{eq:refined_aoi}
\end{equation}
where $J^{MW}$ is the long-term weighted average AoI of the Max-Weight policy.\\
\noindent We define $\Psi_{MW}$ as the empirical intra-frame aging component under the Max-Weight policy. Utilizing the sample mean operator $\bar{\mathbb{M}}[\cdot]$, we express this as:

\begin{align}
    \Psi_{MW} &= \frac{\bar{\mathbb{M}} \Bigg[ \sum_{i=1}^{N} w_i \bigg( \frac{\Delta^2\big(S(t_b)\big) - \Delta\big(S(t_b)\big)}{2} \bigg) \Bigg]}{N \bar{\mathbb{M}}\big[\Delta\big(S(t_b)\big)\big]} \nonumber \\
    &= \frac{\bigg(\sum_{i=1}^{N} w_i\bigg) \cdot \bar{\mathbb{M}}\Bigg[\frac{\Delta^2\big(S(t_b)\big)-\Delta\big(S(t_b)\big)}{2}\Bigg]}{N\bar{\mathbb{M}}\big[\Delta\big(S(t_b)\big)\big]} \nonumber \\
    &= \frac{\sum_{i=1}^N w_i}{2N} \cdot \frac{\bar{\mathbb{M}}\big[\Delta^2\big(S(t_b)\big)\big] - \bar{\mathbb{M}}\big[\Delta\big(S(t_b)\big)\big]}{\bar{\mathbb{M}}\big[\Delta\big(S(t_b)\big)\big]} \nonumber \\
    &= \frac{\sum_{i=1}^N w_i}{2N} \Bigg( \frac{\bar{\mathbb{M}}\big[\Delta^2\big(S(t_b)\big)\big]}{\bar{\mathbb{M}}\big[\Delta\big(S(t_b)\big)\big]} - 1 \Bigg),
    \label{eq:psi_mw_refined}
\end{align}
where $\bar{\mathbb{M}}[X] = \frac{1}{B} \sum_{b=1}^{B} X_b$ denotes the empirical average over a horizon of $B$ scheduling frames.
From (\ref{eq:refined_aoi}) and (\ref{eq:psi_mw_refined}), we get

\begin{equation}
    J^{MW} = L_1 + \Psi_{MW},
\end{equation}
where
\begin{equation}
\begin{split}
    L_1 &= \lim_{B \to \infty} \frac{1}{N \sum_{b=1}^{B} \mathbb{E} \big[ \Delta\big(S(t_b)\big) \big]} \\
    &\quad \times \sum_{b=1}^{B} \mathbb{E} \Bigg[ \sum_{i=1}^{N} w_i A_i(t_b) \Delta\big(S(t_b)\big) \Bigg].
\end{split}
\label{eq:mw_l1}
\end{equation}
From (\ref{eq:J_R_thm}), we get 

\begin{equation}
    \mathbb{E}[J^{R}] = R_1 + \Psi_{SR}, \label{eq:srp_aoi_final}
\end{equation}
where $R_1 = \frac{1}{N} \sum_{i=1}^N \frac{w_i \bar{\Delta}_{SR}}{\gamma_i}$ and $\Psi_{SR} = \frac{1}{N} \sum_{i=1}^{N} w_{i} \left( \frac{\sum_{S} \Delta^{2}(S)\mu^*_{S}}{2 \sum_{S} \Delta(S)\mu^*_{S}} - \frac{1}{2} \right)$.

The network state is analyzed using the quadratic Lyapunov function:
\begin{equation}
    L(t_b) = \sum_{i=1}^N w_i A_i^2(t_b).
\end{equation}
The Max-Weight policy, $\pi^{MW}(t_b)$, minimizes the normalized drift over each frame:
\begin{equation}
    \pi^{MW}(t_b) = \arg \min_{S} \frac{L(t_b+1) - L(t_b)}{\Delta(S)}.
\end{equation}
For any optimal stationary randomized policy $\pi^*$ \cite{neely}:
\begin{equation}
    \frac{\mathbb{E}[a(\pi^*)]}{\mathbb{E}[b(\pi^*)]} \geq \inf_{\pi \in \mathcal{P}} \left[ \frac{\mathbb{E}[a(\pi)]}{\mathbb{E}[b(\pi)]} \right],
\end{equation}
where $\pi^*$ is the optimal SRP and the MW policy provides the bound. Thus, we have:
\begin{equation}
\begin{split}
    &\frac{\mathbb{E} [ L^{MW}(t_b+1) - L^{MW}(t_b) \mid \bar{A}(t_b) ]}{\mathbb{E} [ \Delta^{MW}(t_b) \mid \bar{A}(t_b) ]} \\
    &\quad \leq \frac{\mathbb{E}[L^{SR}(t_b+1) - L^{SR}(t_b) \mid \bar{A}(t_b)]}{\mathbb{E}[\Delta^{SR}(t_b) \mid \bar{A}(t_b)]}.
\end{split}
\end{equation}
Substituting the age evolution from \eqref{eq:aoi_evolution}, the Lyapunov function at $t_b+1$ is:
\begin{align}
    L^{SR}(t_b+1) &= \sum_{i=1}^N w_i \Bigl( A_i(t_b) + \Delta(S(t_b)) \nonumber \\
    &\qquad - A_i(t_b) u_i(S(t_b)) \Bigr)^2 \label{eq:lyapunov_evol} \\
    &= \sum_{i=1}^N w_i \Bigl( A_i^2(t_b) + \Delta^2(S(t_b)) \nonumber \\
    &\qquad - A_i^2(t_b) u_i(S(t_b)) \nonumber \\
    &\qquad + 2 A_i(t_b) \Delta(S(t_b)) \nonumber \\
    &\qquad - 2 A_i(t_b) \Delta(S(t_b)) u_i(S(t_b)) \Bigr). \label{eq:lyapunov_exp}
\end{align}
Now, evaluating the one-slot drift for the stationary randomized policy:
\begin{equation} \label{eq:lyapunov_drift}
    \begin{split}
        &L^{SR}(t_b+1) - L^{SR}(t_b) \\
        &= \sum_{i=1}^N w_i \Bigl( \Delta^2(S(t_b)) - A_i^2(t_b) u_i(S(t_b)) \\
        &\qquad + 2 A_i(t_b) \Delta(S(t_b)) \\
        &\qquad - 2 A_i(t_b) \Delta(S(t_b)) u_i(S(t_b)) \Bigr).
    \end{split}
\end{equation}
Taking the conditional expectation:
\begin{equation}
    \begin{split}
        &\mathbb{E}[L^{SR}(t_b + 1) - L^{SR}(t_b) \mid \bar{A}(t_b)] \\
        &= \sum_{i=1}^N w_i \Bigl( \overline{(\Delta^{SR})^2} - A_i^2(t_b) \mathbb{E}[u_i^{SR}] \\
        &\qquad + 2 A_i(t_b) \mathbb{E}[\Delta^{SR}] - 2 A_i(t_b) \mathbb{E}[(\Delta u_i)^{SR}] \Bigr).
    \end{split}
\end{equation}
\noindent Let 
\begin{align}
    \mathbb{E}[\Delta^{SR}] &= \sum_S \Delta(S) \mu^*_S = \bar{\Delta}^{SR}, \\
    \gamma_i &= \sum_{S: i \in S} \mu^*_S = \mathbb{E}[u_i^{SR}].
\end{align}

\noindent Now, observing that $\Delta \geq 1$, it follows that $\mathbb{E}[(\Delta u_i)^{SR}] \geq \mathbb{E}[(u_i)^{SR}] = \gamma_i$. Substituting these into the drift expression provides the following upper bound:
\begin{align}
    \mathbb{E}[L^{SR}(t_b + 1) &- L^{SR}(t_b) \mid \bar{A}(t_b)] \nonumber \\
    &\le \sum_{i=1}^N w_i \Big( \overline{(\Delta^{SR})^2} + 2 A_i(t_b) \bar{\Delta}^{SR} \nonumber \\
    &\qquad\qquad\quad - A_i^2(t_b) \gamma_i - 2 A_i(t_b) \gamma_i \Big) \label{eq:drift_bound_66} \\
    &= \sum_{i=1}^N w_i \Big( \overline{(\Delta^{SR})^2} + 2 A_i(t_b) \bar{\Delta}^{SR} \nonumber \\
    &\qquad\qquad\quad - A_i(t_b) \big(A_i(t_b) + 2\big) \gamma_i \Big) \label{eq:drift_bound_67}.
\end{align}
By completing the square for $A_i(t_b)$ and normalizing by the expected frame duration, we obtain the full expression for the SRP drift bound:
\begin{equation}
\begin{aligned}
    &\frac{\mathbb{E}[L^{SR}(t_b+1) - L^{SR}(t_b) \mid \bar{A}(t_b)]}{\bar{\Delta}_{SR}} \le \frac{1}{\bar{\Delta}_{SR}} \Bigg[ \\
    &\sum_{i=1}^N -w_i \gamma_i \left( A_i(t_b) - \frac{\bar{\Delta}_{SR}}{\gamma_i} + 1 \right)^2 + G \Bigg],
\end{aligned}
\end{equation}
where the constant scaling term $G$ is defined as:
\begin{equation}
    G = \sum_{i=1}^N w_i \left( \gamma_i \left( \frac{\bar{\Delta}_{SR}}{\gamma_i} - 1 \right)^2 + \overline{(\Delta^{SR})^2} \right).
\end{equation}
Applying the Cauchy-Schwarz inequality to the quadratic age term yields:
\begin{equation}
\begin{aligned}
    &\left( \sum_{i=1}^{N} w_i \gamma_i \left(A_i(t_b) - \frac{\overline{\Delta_{SR}}}{\gamma_i} + 1\right)^2 \right) \left( \sum_{i=1}^{N} w_i \frac{\overline{\Delta_{SR}}}{\gamma_i} \right) \\
    &\geq \left( \sum_{i=1}^{N} w_i \left| A_i(t_b) - \frac{\overline{\Delta_{SR}}}{\gamma_i} + 1 \right| \right)^2 \label{eq:cauchy_schwarz}.
\end{aligned}
\end{equation}

\noindent Taking the conditional expectation with respect to the state $\bar{A}(t_b)$, let $\mathbb{E}[\Delta^{SR}] = \bar{\Delta}^{SR}$ and $\gamma_i = \mathbb{E}[u_i^{SR}]$. Using the property $\mathbb{E}[(\Delta u_i)^{SR}] \geq \gamma_i$, we obtain the upper bound:
\begin{align}
    \mathbb{E}[L^{SR}(t_b + 1) &- L^{SR}(t_b) \mid \bar{A}(t_b)] \nonumber \\
    &\le \sum_{i=1}^N w_i \Big( \overline{(\Delta^{SR})^2} + 2 A_i(t_b) \bar{\Delta}^{SR} \nonumber \\
    &\qquad\qquad\quad - A_i(t_b) \big(A_i(t_b) + 2\big) \gamma_i \Big).
\end{align}

\noindent Completing the square and normalizing by the expected frame duration, the SRP drift bound is:
\begin{equation}
\begin{aligned}
    &\frac{\mathbb{E}[L^{SR}(t_b+1) - L^{SR}(t_b) \mid \bar{A}(t_b)]}{\bar{\Delta}_{SR}} \le \\
    &\frac{1}{\bar{\Delta}_{SR}} \left[ \sum_{i=1}^N -w_i \gamma_i \left( A_i(t_b) - \frac{\bar{\Delta}_{SR}}{\gamma_i} + 1 \right)^2 + G \right],
\end{aligned}
\end{equation}
where $G = \sum_{i=1}^N w_i \left( \gamma_i \left( \frac{\bar{\Delta}_{SR}}{\gamma_i} - 1 \right)^2 + \overline{(\Delta^{SR})^2} \right)$.

The Max-Weight scheduler is designed to minimize the normalized conditional drift at each frame $t_b$. Consequently, the drift under the MW policy is bounded by the normalized drift of the optimal stationary randomized policy (SRP) evaluated at the current state:
\begin{align}
    \label{eq:drift_bound}
    &\frac{\mathbb{E} \Bigl[ L^{MW}(t_{b+1}) - L^{MW}(t_b) \bigm| \bar{A}(t_b) \Bigr]}{\mathbb{E} \Bigl[ \Delta(S_{MW}(t_b)) \bigm| \bar{A}(t_b) \Bigr]} \nonumber \\
    &\quad \le \frac{\mathbb{E} \Bigl[ L^{SR}(t_{b+1}) - L^{SR}(t_b) \bigm| \bar{A}(t_b) \Bigr]}{\mathbb{E} \Bigl[ \Delta(S_{SR}(t_b)) \bigm| \bar{A}(t_b) \Bigr]}.
\end{align}

\noindent Substituting the expected frame duration for the SRP, $\mathbb{E} [ \Delta(S_{SR}(t_b)) \mid \bar{A}(t_b) ] = \bar{\Delta}_{SR}$, and utilizing the drift bound derived earlier through square completion:
\begin{equation}
\begin{aligned}
    &\frac{\mathbb{E} [ L(t_b+1) - L(t_b) \mid \bar{A}(t_b) ]}{\mathbb{E} [ \Delta(S_{MW}(t_b)) \mid \bar{A}(t_b) ]} \le \\
    &\frac{1}{\bar{\Delta}_{SR}} \left[ \sum_{i=1}^N -w_i \gamma_i \left( A_i(t_b) - \frac{\bar{\Delta}_{SR}}{\gamma_i} + 1 \right)^2 + G \right].
\end{aligned}
\end{equation}

\noindent Defining the one-slot conditional Lyapunov drift as $\Delta(\bar{A}(t_b)) \triangleq \mathbb{E} [ L(t_b+1) - L(t_b) \mid \bar{A}(t_b) ]$ and multiplying through by the MW expected frame duration:
\begin{equation}
\begin{aligned}
    \Delta(\bar{A}(t_b)) \le \frac{\mathbb{E} [ \Delta(S_{MW}(t_b)) \mid \bar{A}(t_b) ]}{\bar{\Delta}_{SR}} \times \\ 
    \Bigg[ \sum_{i=1}^N -w_i \gamma_i \left( A_i(t_b) - \frac{\bar{\Delta}_{SR}}{\gamma_i} + 1 \right)^2 + G \Bigg]
\end{aligned}
\end{equation}

\noindent Finally, applying the Cauchy-Schwarz inequality, defining $C_1 = \sum_{i=1}^N \frac{w_i \bar{\Delta}_{SR}}{\gamma_i}$, and rearranging:
\begin{equation}
\begin{aligned}
    &\frac{\mathbb{E} [ \Delta(S_{MW}(t_b)) \mid \bar{A}(t_b) ]}{\bar{\Delta}_{SR}} \left( \sum_{i=1}^{N} w_i \left| A_i(t_b) - \frac{\bar{\Delta}_{SR}}{\gamma_i} + 1 \right| \right)^2 \\ 
    &\leq -C_1 \Delta(\bar{A}(t_b)) + \frac{C_1 G}{\bar{\Delta}_{SR}} \mathbb{E} [ \Delta(S_{MW}(t_b)) \mid \bar{A}(t_b) ].
\end{aligned}
\end{equation}

\noindent Summing over the time horizon $b = 1, \dots, B$ yields:
\begin{equation}
\begin{aligned}
    \sum_{b=1}^{B} &\frac{\mathbb{E} [ \Delta(S_{MW}(t_b)) \mid \bar{A}(t_b) ]}{\bar{\Delta}_{SR}} \left( \sum_{i=1}^{N} w_i \left| A_i(t_b) - \frac{\bar{\Delta}_{SR}}{\gamma_i} + 1 \right| \right)^2 \\
    &\le -C_1 \mathbb{E}[L(t_{B+1}) - L(t_{1})] + \frac{C_1 G}{\bar{\Delta}_{SR}} \sum_{b=1}^{B} \mathbb{E}[\Delta(S_{MW}(t_b))].
\end{aligned}
\end{equation}

\noindent Let $T = \sum_{b=1}^{B} \Delta(S_{MW}(t_b))$ be the total elapsed time. Dividing both sides by the expected total time $\mathbb{E}[T]$:
\begin{equation}
\begin{aligned}
    &\frac{1}{\mathbb{E}[T]} \sum_{b=1}^B \mathbb{E} \left[ \frac{\Delta(S_{MW}(t_b))}{\bar{\Delta}_{SR}} \left( \sum_{i=1}^{N} w_i \left| A_i(t_b) - \frac{\bar{\Delta}_{SR}}{\gamma_i} + 1 \right| \right)^2 \right] \\
    &\le \frac{C_1 \mathbb{E}[L(t_1) - L(t_{B+1})]}{\mathbb{E}[T]} + \frac{C_1 G}{\bar{\Delta}_{SR}}.
\end{aligned}
\end{equation}

\noindent As $B \to \infty$, the total expected duration $\mathbb{E}[T] \to \infty$. For any stable scheduling policy, the normalized telescoping drift vanishes:
\begin{equation}
    \lim_{T \to \infty} \frac{\mathbb{E}[L(t_{B+1}) - L(t_1)]}{\mathbb{E}[T]} = 0.
\end{equation}

\noindent Applying Jensen's inequality to the normalized sum of squared deviations and taking the square root, we bound the time-averaged absolute deviation:
\begin{equation}
    \begin{split}
        \lim_{B \to \infty} &\frac{1}{\mathbb{E}[T]} \mathbb{E} \Biggl[ \sum_{t_b=1}^B \Delta(S_{MW}(t_b)) \\
        &\quad \times \sum_{i=1}^N w_i \left| A_i(t_b) - \frac{\bar{\Delta}_{SR}}{\gamma_i} + 1 \right| \Biggr] \le \sqrt{C_1 G}.
    \end{split}
\end{equation}

\noindent To recover the weighted Age of Information, we decompose the age term using the triangle inequality as follows:
\begin{equation}
    \begin{split}
        \sum_{i=1}^N w_i A_i(t_b) &= \sum_{i=1}^N w_i \biggl( A_i(t_b) - \frac{\bar{\Delta}_{SR}}{\gamma_i} + 1 \biggr) \\
        &\quad + \sum_{i=1}^N w_i \biggl( \frac{\bar{\Delta}_{SR}}{\gamma_i} - 1 \biggr) \\
        &\le \sum_{i=1}^N w_i \biggl| A_i(t_b) - \frac{\bar{\Delta}_{SR}}{\gamma_i} + 1 \biggr| \\
        &\quad + \sum_{i=1}^N \frac{w_i \bar{\Delta}_{SR}}{\gamma_i} - \sum_{i=1}^N w_i.
    \end{split}
\end{equation}

\noindent By aggregating these costs across all frames, normalizing by the total duration, and taking the limit, the final performance bound of the Max-Weight policy is expressed as:
\begin{equation} 
    \lim_{B \to \infty} \frac{\sum_{b=1}^{B} \mathbb{E} \left[ \Delta(S_{MW}(b)) \sum_{i=1}^{N} w_i A_i(b) \right]}{\sum_{b=1}^{B} \mathbb{E} [\Delta(S_{MW}(b))]} \leq C_1 + \sqrt{C_1 G}. \label{eq:final_bound}
\end{equation}

\noindent Using the approximations $\frac{\bar{\Delta}_{SR}}{\gamma_i} - 1 \le \frac{\bar{\Delta}_{SR}}{\gamma_i}$ and $\bar{\Delta}_{SR}^2 \le \frac{\bar{\Delta}_{SR}^2}{\gamma_i}$:
\begin{equation}
    G \le \sum_{i=1}^N w_i \left( \gamma_i \left( \frac{\bar{\Delta}_{SR}}{\gamma_i} \right)^2 + \frac{\bar{\Delta}_{SR}^2}{\gamma_i} \right) = \sum_{i=1}^N 2 w_i \frac{\bar{\Delta}_{SR}^2}{\gamma_i}.
\end{equation}

\noindent Substituting $C_1 = \sum_{i=1}^N \frac{w_i \bar{\Delta}_{SR}}{\gamma_i}$ and the bound for $G$ into the square root term:
\begin{equation}
\begin{aligned}
    \sqrt{C_1 G} &\le \sqrt{ \left( \sum_{i=1}^N \frac{w_i \bar{\Delta}_{SR}}{\gamma_i} \right) \left( 2 \bar{\Delta}_{SR} \sum_{i=1}^N \frac{w_i \bar{\Delta}_{SR}}{\gamma_i} \right) } \\
    &= \sqrt{2 \bar{\Delta}_{SR}} \times \bigg( \sum_{i=1}^N \frac{w_i \bar{\Delta}_{SR}}{\gamma_i} \bigg).
    \label{eq:intermediate_bound}
\end{aligned}
\end{equation}
Now, combining (\ref{eq:final_bound}) and (\ref{eq:intermediate_bound}), we have 

\begin{equation}
    L_1 \leq \sqrt{2 \bar{\Delta}_{SR}} \cdot R_1.
    \label{eq:intermediate_bound_2}
\end{equation}

\noindent We need the optimality ratio,
\begin{equation}
\frac{J^{MW}}{\mathbb{E}[J^{R}]} = \frac{L_1 + \Psi_{MW}}{R_1 + \Psi_{SR}} .  
\end{equation}
Note that: 
\begin{equation}
    \frac{L_1 + \Psi_{MW}}{R_1 + \Psi_{SR}} = \frac{L_1}{R_1 + \Psi_{SR}} + \frac{\Psi_{MW}}{R1 + \Psi_{SR}} \leq \frac{L_1}{R_1} + \frac{\Psi_{MW}}{L_B},
    \label{eq:aoi_ratio}
\end{equation}
since we have $R_1 \leq R_1 + \Psi_{SR}$ and $L_B \leq \mathbb{E}[J^{R}] = R_1 + \Psi_{SR}$ by definition.
Using (\ref{eq:intermediate_bound}) and simplifying,
\begin{equation}
    \frac{J^{MW}}{\mathbb{E}[J^{R}]} \leq \sqrt{2 \bar{\Delta}_{SR}} + \frac{\Psi_{MW}}{L_B}.
\end{equation}
Using (\ref{eq:opt_sr}) and simplifying, we arrive at the final optimality ratio:

\begin{equation}
    \rho^{MW} \le  \bigg(2 + \frac{\Psi_{LB}}{L_B}\bigg)\bigg(\sqrt{2 \bar{\Delta}_{SR}} + \frac{\Psi_{MW}}{L_B}\bigg).
\end{equation}

\section{Proof of Corollary \ref{corollary: approx_optimality}}
\label{proof: approx_algo}

The proof follows the same steps as proof of Theorem \ref{thm:mw_optimality_ratio} in Appendix \ref{proof:thm_mwopt}. We start with the normalized conditional drift inequality in (\ref{eq:drift_bound}). From Lemma \ref{lemma:frame_based_optimality}, the approximation algorithm identifies a schedule $\pi^{AMW}$ such that 
\begin{equation}
    \begin{split}
        &\frac{L^{AMW}(t_{b+1}) - L^{AMW}(t_{b})}{\Delta(\pi^{AMW}(t_b))} \\
        &\quad \leq (4+\epsilon) \frac{L^{MW}(t_{b+1}) - L^{MW}(t_{b})}{\Delta(\pi^{MW}(t_b))}
    \end{split}
\end{equation}
where $\pi^{MW}$ is the optimal solution to \eqref{eqn:mw_policy}.

By substituting this approximation guarantee into the drift bound, the inequality in (76) is scaled by a factor of $(4+\epsilon)$, i.e.,
\begin{align}
    \label{eqn:drift_bound_amw}
    &\frac{\mathbb{E} \Bigl[ L^{AMW}(t_{b+1}) - L^{AMW}(t_b) \bigm| \bar{A}(t_b) \Bigr]}{\mathbb{E} \Bigl[ \Delta(S_{AMW}(t_b)) \bigm| \bar{A}(t_b) \Bigr]} \nonumber \\
    &\le (4+\epsilon) \cdot \frac{\mathbb{E} \Bigl[ L^{SR}(t_{b+1}) - L^{SR}(t_b) \bigm| \bar{A}(t_b) \Bigr]}{\mathbb{E} \Bigl[ \Delta(S_{SR}(t_b)) \bigm| \bar{A}(t_b) \Bigr]}.
\end{align}

Following the steps in Appendix \ref{proof:thm_mwopt}, we arrive at (\ref{eq:intermediate_bound_2}) with a scaled factor:
\begin{equation}
    L_1 \leq (4+\epsilon)\cdot\sqrt{2 \bar{\Delta}_{SR}} \cdot R_1.
    \label{eq:intermediate_bound_amw}
\end{equation}

Again, following the steps, (\ref{eq:aoi_ratio}) becomes
\begin{equation}
    \frac{J^{AMW}}{\mathbb{E}[J^{R}]} \leq (4+\epsilon)\cdot\sqrt{2 \bar{\Delta}_{SR}} + \frac{\Psi_{MW}}{L_B},
\end{equation}
where $J^{AMW}$ is the long-term weighted average AoI of the Approximate Max-Weight (AMW) policy. 
Using (\ref{eq:opt_sr}) and simplifying, we obtain the optimality ratio in (\ref{eq:corollary}).
\qed
\end{document}